\documentclass[10pt,a4paper]{article}

\usepackage{authblk}
\usepackage{microtype}
\usepackage{graphicx}
\usepackage{subfigure}
\usepackage{booktabs}
\usepackage[square,numbers]{natbib}
\usepackage{hyperref}
\usepackage{a4wide}

\usepackage{amsmath}
\usepackage{amssymb}
\usepackage{mathtools}
\usepackage{amsthm}
\usepackage[capitalize,noabbrev]{cleveref}

\theoremstyle{plain}
\newtheorem{theorem}{Theorem}[section]
\newtheorem{proposition}[theorem]{Proposition}
\newtheorem{lemma}[theorem]{Lemma}

\theoremstyle{definition}

\newtheorem{example}{Example}
\theoremstyle{remark}

\usepackage{nicefrac}
\usepackage{enumitem}
\usepackage{authblk}
\usepackage{algorithm}
\usepackage{algorithmic}

\usepackage{pgfplots}
\usepackage{pgfplotstable}
\pgfplotsset{compat=1.14}

\renewcommand{\algorithmicrequire}{\textbf{Input:}}
\renewcommand{\algorithmicensure}{\textbf{Output:}}
\begin{document}

\title{Fair Indivisible Payoffs through Shapley Value}

\author[1]{Miko{\l}aj Czarnecki}
\author[1,2]{Micha{\l} Korniak}
\author[1]{Oskar Skibski}
\author[1]{Piotr~Skowron}

\affil[1]{\normalsize University of Warsaw, Banacha 2, 02-097 Warszawa, Poland}
\affil[2]{\normalsize ETH Z\"urich, Rämistrasse 101, 8092 Z\"urich, Switzerland}
\affil[ ]{\textit{mc448206@students.mimuw.edu.pl},\break
\textit{mkorniak@ethz.ch},\break
\textit{\{oskar.skibski, p.skowron\}@mimuw.edu.pl}}

\maketitle

\begin{abstract}
We consider the problem of payoff division in \emph{indivisible coalitional games}, where the value of the grand coalition is a natural number.
This number represents a certain quantity of indivisible objects, such as parliamentary seats, kidney exchanges, or top features contributing to the outcome of a machine learning model.
The goal of this paper is to propose a fair method for dividing these objects among players.
To achieve this, we define the \emph{indivisible Shapley value} and study its properties.
We demonstrate our proposed technique using three case studies, in particular, we use it to identify key regions of an image in the context of an image classification task.
% We demonstrate our proposed technique using three case studies. Specifically, we apply it as an apportionment method in approval elections, use it to identify key regions of an image in the context of an image classification task, and employ it to define a fair allocation of seats for a coalition in a parliamentary election.
\end{abstract}

\section{Introduction}\label{introduction}

We consider a model in which a group of players seeks to split a fixed number of indivisible objects among themselves. 
For each subgroup of players, hereinafter called a \emph{coalition}, the exact or estimated number of objects they are entitled to is known. 
The goal of this paper is to develop a method for identifying a fair division while accounting for dependencies between players.

Our setting models various scenarios. E.g., in kidney exchange programs~\cite{Agarwal:etal:2019, Chandra:Skinner:2012}, hospitals can be considered as players, with coalitions' values representing the number of kidney exchanges they can perform through collaboration. The outcome of our algorithms can be interpreted as a recommended distribution of donors across different hospitals for a single round of the allocation process.

Similarly, our algorithms can be used as apportionment methods for elections, where voters are allowed to approve multiple political parties~\cite{Brill:etal:2024, Brams:etal:2019}. In this context, a coalition's value would represent the number of seats it is entitled to, calculated as the percentage of votes cast on any subset of parties from the group, multiplied by the total number of seats. Our algorithms then determine a seat distribution among parties that can be considered fair both for voters, and for the participating parties.

Yet another application involves selecting the most important features contributing to the outcome of a machine learning model~\cite{Lundberg:Lee:2017,KommiyaMothilal:etal:2021}. 
For example, consider an image classification task where an image is divided into multiple small regions that represent features for the model. 
By evaluating all subsets of features, one can assess the contribution of each feature to the final outcome. 
Our method identifies a few key regions that represent the distribution of importance across the entire image.  

We frame our problem as a coalitional game where the value of the coalition of all players is a natural number. 
Depending on the application, the values of individual coalitions may or may not be integral. 
To achieve a fair division, we introduce the \emph{indivisible Shapley value}---
an extension of the classic Shapley value that returns a vector of integer payoffs. 
Our method accounts for dependencies between players, and so prevents the overrepresentation of players that share common resources.  

To illustrate the idea of fairness, consider an apportionment scenario with three right-wing parties, $A$, $B$, and $C$, and two left-wing parties, $D$ and $E$.  Suppose there are three committee seats to distribute. The three right-wing parties are approved by 66\% of voters and thus entitled to two out of three seats, while the two left-wing parties are approved by 33\% of voters and entitled to one seat. 
Clearly, each right-wing party is entitled to 2/3 of a seat, while $D$ and $E$ are entitled to half a seat, each.
If seats were assigned by simply rounding up the largest shares, all three seats would go to the right-wing parties, leaving 33\% of the electorate underrepresented.
Our method, however, would allocate two seats randomly among the right-wing parties and one seat to a left-wing party, resulting in an arguably more balanced and representative outcome. 

In coalitional games, the concept of a justified payoff for all groups is captured by the \emph{core}. 
We show that in convex indivisible games with integer coalition values---where it is known that the Shapley value lies in the core---the indivisible Shapley value also belongs to the core. 
However, we demonstrate that such guarantees do not extend to general fractional or non-convex games.

In our work we chose to focus on a deterministic algorithm. It yields desirable outcomes in settings such as elections where transparency and predictability are desired. Likewise, in machine-learning context determinism guarantees reproducibility and consistency. 

Additionally, we consider a setting where objects have the same value but are non-identical. 
We develop an algorithm to find an allocation corresponding to the indivisible Shapley value. 
Finally, we demonstrate our approach using three case studies: applying it as an apportionment method in approval elections, selecting key regions in an image classification task, and as a tool for a fair division of seats among political parties forming a coalition.

\textbf{Related Work}
The fair allocation of indivisible goods has been studied extensively in the literature. Most previous works conceptualize fairness through envy-freeness~\cite{Maskin1987, Segalevi2019}, emphasizing individual-level comparisons. In contrast, we study fairness within the framework of coalitional games, which focuses on relationships and stability among groups of agents rather than individuals.

Coalitional games are often defined using integer representations, for example in weighted voting games~\cite{Chalkiadakis:etal:2011,Shapley:Shubik:1954}, matching games~\cite{Benedek:etal:2023}, and allocation games~\cite{greco2015}. These formulations frequently employ integer parameters such as weights or capacities. However, the players’ payoffs or marginal contributions are not restricted to integer values, and the resulting Shapley values are typically non-integral. Our work differs in that we focus on coalitional games where both the representation and the resulting payoffs are explicitly integer-valued.

Indivisible coalitional games can be viewed as a subset of Non-Transferable Utility (NTU) games~\cite{mclean2002ntu}. However, it remains unclear whether previously proposed solution concepts~\cite{lejano2011} yield fair or meaningful outcomes when applied specifically to indivisible games with integer payoffs. In particular, the Shapley NTU value, as introduced by McLean~\cite{mclean2002ntu}, is undefined in this setting, since no value exists that satisfies the Unanimity axiom (i.e., equally dividing the payoff in a unanimity game).
% To the best of our knowledge, our paper is the first one that studies coalitional games with the restriction of the set of outcomes to integer values.
% Notably, coalitional games are often concisely defined through integers, as for example in the class of weighted voting games~\cite{Chalkiadakis:etal:2011,Shapley:Shubik:1954} or matching games~\cite{Benedek:etal:2023}.
% However, in such games, the values of players are not restricted to integers and the Shapley values are typically non-integral.

% We discuss the literature related to the specific applications in the relevant sections when discussing each respective case study.

\section{Preliminaries}\label{section:preliminiaries}
A (coalitional) \emph{game} is a pair $(N,v)$ such that $N$ is the set of $n$ players and $v: 2^N \rightarrow \mathbb{R}$ is a function that assigns a real value to every coalition $S \subseteq N$ with $v(\emptyset) = 0$.

A \emph{value} of a game is a function $\varphi$ that in every game $(N,v)$ assigns a payoff vector $x \in \mathbb{R}^N$; the payoff of player $i$ according to $\varphi$ in game $(N,v)$ is denoted by $\varphi_i(N,v)$.
In this paper, we will focus on the \emph{Shapley value}~\cite{Shapley:1953}.
Let $\pi$ be an arbitrary permutation from $\Omega(N)$, the set of all permutations of $N$, and let $S^{\pi}_i = \{j \in N : \pi(j) < \pi(i)\}$ be the set of players that appear before $i$ in $\pi$.
The Shapley value of player $i$ in game $(N,v)$ is equal to the average marginal contribution to all permutations:
\[ SV_i(N,v) = \frac{1}{|N|!} \sum_{\pi \in \Omega(N)} (v(S^{\pi}_i \cup \{i\}) - v(S^{\pi}_i)). \]

\citet{Shapley:1953} in his classic paper characterized the Shapley value with four axioms that capture the fairness of the payoff division.
We list them here, as they will play a crucial role both in the analysis and the proofs.

\begin{description}[itemsep=1pt,parsep=1pt,topsep=1pt]
\item[Efficiency:] $\sum_{i \in N} \varphi_i(N,v) \!=\! v(N)$ for every game $(N,v)$
\item[Null-player:] $\varphi_i(N,v) = 0$ for every game $(N,v)$ and player $i \in N$ which is a null-player, i.e., $v(S \cup \{i\}) = v(S)$ for every $S \subseteq N \setminus \{i\}$
\item[Additivity:] $\varphi_i(N,v+v') = \varphi_i(N,v) + \varphi_i(N,v')$ for every games $(N,v), (N,v')$ and $i \in N$, where $(v+v')(S) = v(S)+v'(S)$ for every $S \subseteq N$
\item[Symmetry:] $\varphi_i(N,v) = \varphi_j(N,v)$ for every game $(N,v)$ and players $i,j \in N$ such that $v(S \cup \{i\}) = v(S \cup \{j\})$ for every $S \subseteq N \setminus \{i,j\}$
\end{description}

Efficiency states that payoffs of all players sum up to $v(N)$.
Null-player implies that a null-player get zero payoff.
Additivity means the sum of payoffs in two games is equal to the payoff in the game obtained by summing both games.
Symmetry states symmetric players gets equal payoffs. 

Among many characterizations of the Shapley value, one of the most important was proposed by \citet{Harsanyi:1963}.
The \emph{Harsanyi dividend} of coalition $S$ in game $(N,v)$ is the surplus created by a coalition $S$. 
It is defined as the value of the coalition decreased by surpluses of subcoalitions:
\[ \Delta_v(S) = v(S) - \sum_{T \subsetneq S} \Delta_v(T) \quad \text{ if }|S|>1, \]
with $\Delta_v(S) = v(S)$ for $|S|=1$.
For a coalition $S \subseteq N$, let us define an \emph{unanimity game} $u_S$ in which all supersets of $S$ have value $1$, and the remaining coalitions have value $0$: $u_S(T) = 1$ if $S \subseteq T$ and $u_S(T) = 0$, otherwise.
Now, it holds that every game can be represented as:
$v = \sum_{S \subseteq N} u_S \Delta_v(S)$.
From Efficiency, Symmetry and Null-player it is easy to observe that the Shapley values of all players from $S$ in $u_S$ equal $1/|S|$ and of others are zero.
Hence, from Additivity, we get that the Shapley value simply shares the dividend of each coalition equally among its members:
\[ SV_i(N,v) = \sum_{S \subseteq N, i \in S} \Delta_v(S)/|S|. \]

The Shapley value focuses on the fairness of the division. 
In turn, the \emph{core} is a set of \emph{stable} payoff vectors:
\[ \mathop{Core}(N,v) = \{x : \sum_{i \in N} x_i = v(N), \forall_{S \subseteq N} \sum_{i \in S} x_i \ge v(S)\}. \]
Hence, the payoff vector is in the core if the payoff of every coalition is larger or equal to its value.

In general, the core can be empty. 
However, it is known that the core is not empty if the game is \emph{convex}.
A game is \emph{convex} if it holds
$v(S \cup \{i\}) - v(S) \le v(T \cup \{i\}) - v(T)$
for every $i \in N, S \subseteq T \subseteq N \setminus \{i\}$.
Roughly speaking, game is convex if the marginal contribution of a player to a coalition increases (does not decrease) with the coalition size.
If a game is convex, then the Shapley value belongs to the core.
A subclass of convex games is formed by \emph{positive} games.
A game is \emph{positive} if all its Harsanyi dividends are non-negative: $\Delta_v(S) \ge 0$ for all $S \subseteq N$.

\section{Indivisible Coalitional Games}\label{section:model}

In this paper, we focus on a new class of coalitional games, named \emph{indivisible (coalitional) games}. 
In such games, the value of the grand coalition represents the number of indivisible identical items (e.g., seats in parliament, kidney transplants). 
Formally, an \emph{indivisible coalitional game} is a coalitional game $(N,v)$ with $v(N) \in \mathbb{N}$. 

As we argued in the introduction, depending on the application, values of all coalitions may or may not be integers as well.
In the former case, i.e., if $v(S) \in \mathbb{Z}$ for every $S \subseteq N$, we will call such games \emph{integer indivisible games}.
In the later case, we will call such games \emph{fractional indivisible games}.

The value of an indivisible coalitional game is an \emph{indivisible payoff vector} $x \in \mathbb{Z}^N$ that assigns each player an integer.
The goal of this paper is to extend the Shapley value to indivisible coalitional games.

From Shapley's axioms, we cannot demand Symmetry and Additivity.
This is clear by looking at game $u_{\{1,2\}}$ in which two players jointly create a value of one, but only one of them can receive payoff $1$. 
If we assume that player $1$ receives this payoff, then Additivity would imply that in game $2 u_{\{1,2\}}$ player $1$ will receive the payoff of $2$, which is clearly undesirable.

The following two properties captures the idea we want the value to be as closed to the Shapley value as possible.
Specifically, the Lower Quota states that the payoff of a player is no less than the floor of its Shapley value and the Upper Quota---that the payoff is no greater than the ceiling.
\begin{description}
\item[Lower Quota:] $\varphi_i(N,v) \ge \lfloor SV_i(N,v) \rfloor$ for every game $(N,v)$ and player $i \in N$.
\item[Upper Quota:] $\varphi_i(N,v) \le \lceil SV_i(N,v) \rceil$ for every game $(N,v)$ and player $i \in N$.
\end{description}
Note that the Lower Quota and the Upper Quota imply Null-player.

We define \emph{indivisible core} as the set of indivisible payoff vectors that belong to the core.
We note that the indivisible core might be empty even if the standard core is not; in fact, it might be empty even for integer indivisible games.

\begin{example}
\label{example:half-game}
\emph{(Half-game)}
Consider an integer indivisible game with $N = \{1,2,3,4\}$ and $v(S) = \lfloor |S|/2 \rfloor$.
Note that $v(N) = 2$.
Fix an arbitrary payoff vector $x = [x_1,\dots,x_4] \in \mathbb{R}^N$ and take two players with the smallest payoffs: $i,j \in N$.
Since $v(\{i,j\}) = 1$ we know that $x_i+x_j \ge 1$. 
This is only possible if $x = [0.5, 0.5, 0.5, 0.5]$.
Thus, the core consists of only this one payoff vector and the indivisible core is empty.
\hfill $\lrcorner$
\end{example}

\section{Indivisible Shapley Value for Convex Integer Games}\label{section:convex}

We begin our analysis from convex integer games. 
As we already mentioned, we know that the core of a convex game is non-empty and that the Shapley value belongs to it \cite{Shapley:1971}.
However, as evidence by \cref{example:half-game}, this does not imply that any indivisible payoff vector is in the core.

% In fact, it is easy to find an indivisible payoff vector which belongs to the core.
% Let $\pi$ be an arbitrary permutation and consider a payoff vector $x_i = v(S^{\pi}_i \cup \{i\})-v(S^{\pi}_i)$.
% Since the values of coalitions are integers, $x$ is indivisible payoff vector and it is known that it belongs to the core~\cite{Shapley:1971}.
% However, this vector may violate both Lower Quota and Upper Quota and it is clearly unfair. In particular, if all $n$ players are necessary to create $k$ goods, e.g., in game $k \cdot u_N$, the method will give the whole payoff to the last player in the permutation.

In this section, we show that there exists an indivisible payoff vector that satisfies Lower Quota and Upper Quota and at the same time belongs to the core and we define a method that finds such a vector.

\subsection{Games with Shapley Values from the Unit Interval}

Let us first focus on games in which Shapley values of all players belong to $(0,1)$ interval.
From Efficiency, this implies that the value of the grand coalition is lower than the number of players: $v(N) < |N|$.
Hence, to satisfy both Lower and Upper Quota, we need to select a subset of players that will receive value $1$ and give $0$ to the remaining players.

At first, it seems natural to select players with the highest Shapley values. 
However, as we now argue, this is not a desirable approach.

\begin{example}\label{example:2u123u45}
Consider five players, three of which are responsible for creating two goods, and the remaining two -- for another good. 
This corresponds to a game $2u_{123}+u_{45}$.
The Shapley value of players $1$, $2$ and $3$ equals $0.\bar{6}$ and of players $4$ and $5$ -- only $0.5$.
Hence, by choosing players with the highest Shapley values we would select players $1$, $2$ and $3$. 
This is, however, unfair to players $4$ and $5$ who are responsible for one good: the payoff vector $(1,1,1,0,0)$ is not in the core, as $x_4+x_5 < v(\{4,5\})$.
\hfill $\lrcorner$
\end{example}

Let $(N,v)$ be any integer indivisible game such that $SV_i(N,v) \in (0,1)$ for every $i \in N$.
We already argued that $v(N) < |N|$.
From the fact that the Shapley value belongs to the core $v(S) \le \sum_{i \in S} SV_i(N,v)$, we get that the same bound $v(S) < |S|$ holds for all coalitions $S \subseteq N$.
% \[ v(S) \le \sum_{i \in S} SV_i(N,v) < |S| \quad \mbox{ for every }S \subseteq N, S \neq \emptyset. \]
In such a case, we will say that a game is \emph{size-bounded}.

Our construction of the solution will be based on the notion of \emph{reduced game}~\cite{Davis:Maschler:1965}.
Let $(N,v)$ be an arbitrary game and $i \in N$ be a player.
For an arbitrary payoff $c \in \mathbb{R}_{\ge 0}$, the \emph{$c$-reduced game} is $(N \setminus \{i\}, \Psi^{i\rightarrow c}_v)$, where $\Psi^{i\rightarrow c}_v$ is defined as follows:
\[ \Psi^{i\rightarrow c}_{v}(S) = \begin{cases}
v(N)-c & \mbox{if } S = N \setminus \{i\}, \\
0 & \mbox{if } S = \emptyset, \\
\max\{v(S \cup \{i\}) - c,v(S)\} & \mbox{otherwise.} 
\end{cases}\]
The reduced game imitates giving player $i$ the payoff $c$ and removing it from the game. 
Specifically, the value of the grand coalition in the resulting game is reduced by $c$ and, as always, the value of an empty coalition is zero.
For all other coalitions we take the maximum value from the value of a coalition without player $i$ and the value of a coalition with player $i$ reduced by its payoff, $c$.

\begin{example} 
Consider the game $v = 2u_{123}+u_{45}$ from \Cref{example:2u123u45} and assume that we assign payoff $1$ to player $1$.
The reduced game equals $\Psi^{1 \rightarrow 1}_v = u_{23}+u_{45}$.
This is because one good created by a large group has already been allocated.
If we allocate the second good to player $2$ and reduce the game, we will get $\Psi^{2 \rightarrow 1}_{\Psi^{1 \rightarrow 1}_v} = u_{45}$.
Hence, players 4 and 5 would share the last good.
\hfill $\lrcorner$
\end{example}

We begin by showing how the properties of a game translate to the properties of the reduced game.
Specifically, if game is convex and $c$ is between the marginal contributions to the largest and smallest coalitions, the reduced game will also be convex.
Also, if the game is size-bounded, it will remain size-bounded if $c \ge 1$.

\begin{lemma}
\label{lemma:reduced_game}
For every integer indivisible game $(N,v)$, player $i \in N$ and $c \in \mathbb{N}$, the $c$-reduced game $(N \setminus \{i\}, \Psi^{i \rightarrow c}_{v})$ is an integer indivisible game which is
\begin{itemize}[itemsep=1pt,parsep=1pt,topsep=1pt]
\item convex if $(N,v)$ is convex and $v(\{i\}) \le c \le v(N)-v(N \setminus \{i\})$;
\item size-bounded if $(N,v)$ is size-bounded and $c \ge 1$.
\end{itemize}
\end{lemma}

This result implies that if player has a positive marginal contribution to the grand coalition (which, for integers, means equal at least $1$) in a convex, size-bounded game, the $1$-reduced game will remain convex and size-bounded.

Consider now a payoff vector which belongs to the core of the reduced game.
Observe that the value of each coalition is at most $c$ lower than the value of this coalition with player~$i$ added in the original game.
Hence, if we add player~$i$ with value $c$ to this payoff vector, then it will be in the core of the original game.
This leads to the following result.

\begin{lemma}
\label{lemma:in-core-iff-marg-1}
For every convex and size-bounded integer game $(N,v)$ and player $i \in N$, there exists a 0-1 payoff vector $x \in \{0,1\}^N$ in the core such that $x_i=1$ if and only if $v(N) - v(N \setminus \{i\}) \ge 1$. 
\end{lemma}

Our analysis indicates the following scheme for distributing the value of the grand coalition in convex integer games.  
We select a player whose marginal contribution to the grand coalition equal at least $1$.  
We assign $1$ to the selected player and create a reduced game.  
Based on \Cref{lemma:reduced_game}, we know that the resulting game will also be a convex, size-bounded integer indivisible game.  
Thus, we repeat this step iteratively until the value of the grand coalition becomes zero.

It remains to specify the order in which we consider players.  
Among the players with a positive marginal contribution to the grand coalition, we will always choose the remaining player with the highest Shapley value in the original game.  
In this way, while in many games players with the largest remainders will receive the additional units, we ensure that the payoff vector is in the core.

\begin{algorithm}[t]
   \caption{Indivisible Shapley Value}
   \label{alg:main}
\begin{algorithmic}[1]
   \REQUIRE An indivisible game $(N, v)$
   \ENSURE A payoff vector $x \in \mathbb{Z}^N$
   \STATE $\pi \leftarrow$ permutation of players in decreasing order according to $SV_i(N,v) - \lfloor SV_i(N,v) \rfloor$.
   \STATE $\varphi_i \leftarrow \lfloor SV_i(N, v) \rfloor$ for every $i \in N$
   \STATE $v(S) \leftarrow v(S) - \sum_{i \in S} \lfloor SV_i(N,v) \rfloor $ for each $S \subseteq N$
   \STATE $(M,u) \leftarrow (N,v)$
   \FOR {$i \in N$ s.t. $SV_i(N,v) = \lfloor SV_i(N,v) \rfloor$}
   \STATE $(M,u) \leftarrow (M \setminus \{i\}, \Psi^{i \rightarrow 0}_{u})$
   \ENDFOR
   \STATE \colorbox{white!88!blue}{$u(S) \leftarrow \lfloor u(S) \rfloor$ for every $S \subseteq M$}

   \WHILE{$u(M) > 0$}
   \STATE $i \leftarrow$ first player in $\pi$ from $M$ s.t. $u(M) \ge u(M \setminus \{i\})+1$ \colorbox{white!88!blue}{ or $SV_i(M,u) > 0$}
   \STATE $(M, u) \leftarrow (M \setminus \{i\}, \Psi^{i \rightarrow 1}_{u})$
   \STATE $\varphi_i \leftarrow \varphi_i + 1$
   \ENDWHILE

   \STATE {\bfseries return} $\varphi$
\end{algorithmic}
\end{algorithm}
\subsection{General Convex Integer Games}

To apply our method to general convex integer indivisible games, we need to perform two additional steps.  
First, we assign to each player the floor of their Shapley value, $\lfloor SV_i(N,v) \rfloor$, and update the game accordingly:  
$u(S) = v(S) - \sum_{i \in S} \lfloor SV_i(N,v) \rfloor$.
Clearly this game remains convex.
%: the marginal contribution of each player to all coalitions decreases by the same amount, so it is still larger for larger coalitions.  

Second, to ensure that players with zero Shapley value in the resulting game do not receive more than zero, we remove them from the game.  
This must be done carefully, using the reduced game; otherwise, the payoff vector we determine may not belong to the core of the original game.  
%By using the reduced game, we know that assigning $0$ payoffs to the removed players keeps the payoff vector within the core.  

\Cref{alg:main} presents the full procedure. 
Text highlighted in blue marks the changes we will introduce in the next section that do not affect the procedure if the input game is a convex integer game.
Furthermore, in the following theorem, we prove that the resulting payoff vector is the closest indivisible payoff vector to the Shapley value that belongs to the core.  

\begin{theorem}
\label{theorem:convex}
    For convex integer games, the indivisible Shapley value:
	\begin{enumerate}[label=(\alph*),itemsep=1pt,parsep=1pt,topsep=1pt]
        \item is an indivisible payoff vector;
        \item satisfies Efficiency, Lower Quota and Upper Quota;
        \item belongs to the core of game $(N,v)$; and
        \item has a minimal distance to the Shapley value out of all payoff vectors satisfying (a)-(c).
    \end{enumerate}
\end{theorem}

\section{Indivisible Shapley Value for General Games}\label{section:general_games}

Let us discuss how to generalize our method for all games.
We begin with general convex fractional games.  
As we have argued, the core of such games is non-empty and contains the Shapley value.  
However, in contrast to convex integer games, it might happen that no indivisible payoff vector lies in the core.  
In fact, indivisible payoff vectors might be arbitrarily far from the core.

\begin{example}
\label{example:fractional-half-game}
\emph{(Fractional-Half-Game)}  
Consider an integer indivisible game with $N = \{1, \dots, n\}$, where $n$ is even, and $v(S) = |S|/2$.  
We have $v(N) = n/2$, and clearly $x = [0.5, \dots, 0.5]$ is the only payoff vector in the core.  
However, in every indivisible payoff vector, there exist at least $n/2$ players with a payoff of $0$.  
Thus, for the coalition $S$ of these players, we have $\sum_{i \in S} x_i \!=\! 0 \!<\! n/4 \!=\! v(S)$.  
\hfill $\lrcorner$  
\end{example}

To adapt our method to general fractional games, we introduce two changes (marked in blue in \Cref{alg:main}).  
First, after assigning the floor of the Shapley value to all players and removing players with zero Shapley value from the resulting game, we round down all the values in the game (line 8).  
Second, we modify the condition used to select the player who receives an additional $1$ in the final loop (line 10).  
See \cref{appendix:general_games} for details.
We note that both modifications do not change how the procedure works in the case of convex integer games.  
We also prove that \Cref{alg:main} in general games also satisfies Efficiency, Lower Quota and Upper Quota.

\section{ISV-Approach for Large Games}\label{section:heuristic}

The indivisible Shapley value operates under the assumption that the full coalitional game is given as input. However, since the size of the game grows exponentially with respect to the number of players, this assumption is not realistic. In this section, we present an algorithm that simulate the behavior of the indivisible Shapley value: approximates the Shapley value and then operates in polynomial time to determine the final outcome.

\begin{algorithm}[t!]
   \caption{ISV-Approach for Large Games}
   \label{alg:heuristic}
   
\begin{algorithmic}[1]
\REQUIRE An indivisible game \((N,v)\), where \(v\) is a black box
\ENSURE A payoff vector \(x \in \mathbb{Z}^N\)
   \STATE \(x_i \leftarrow 0\) for every \(i \in N\)
   \STATE \((\varphi_i)_{i \in N} \leftarrow\) approx. of the Shapley values
   \STATE \((\varphi_{ij})_{i,j \in N} \leftarrow\) approx. of the Shapley value matrix
   \WHILE{\(v(N) > \sum_{i \in N} x_i\)}
      \STATE \(i \leftarrow\) player with the highest value \(\varphi_i\)
      \STATE \(x_i \leftarrow x_i + 1\)
      \IF{\(\varphi_i > 1\)}
         \STATE \(\varphi_i \leftarrow \varphi_i - 1\)
      \ELSE
      	\STATE \(\varphi_j \leftarrow \varphi_j - \frac{(1 - \varphi_i) \varphi_{ij}}{{\sum_{k \in N \setminus \{i\}} \varphi_{ik}}}\) for every \(j \in N \setminus \{i\}\)
      	\STATE \(\varphi_i \leftarrow 0\)
      \ENDIF
   \ENDWHILE
   \STATE {\bfseries return} \(x\)
\end{algorithmic}
\end{algorithm}

We utilize the notion of the Shapley value matrix $(SV_{ij}(N,v))_{i,j \in N}$~\cite{Hausken:Mohr:2001}. 
In this matrix, the Shapley value of a player is decomposed into a vector $(SV_{i1}(N,v), \dots, SV_{in}(N,v))$, where each element represents the synergy between two players: the contribution of player \(j\) to the value of player \(i\):
\begin{equation}\label{eq:shapley_matrix}
SV_{ij}(N,v) = \sum_{i,j \in S \subseteq N} \frac{\Delta_v(S)}{|S|^2}.
\end{equation}
Clearly, \(SV_i(N,v) = \sum_{j \in N} SV_{ij}(N,v)\).
To avoid exponential blow-up, we rely only on the Shapley values and the Shapley value matrix as sources of information about the game and the dependencies between players.

Our procedure operates as follows. 
First, we approximate the Shapley values and the Shapley value matrix by sampling random permutations.
Then, in a loop, we select the player with the highest Shapley value and add 1 to this player's payoff.
To account for this increase, we decrease the player's Shapley value by 1.
However, if the Shapley value is smaller than 1, we also reduce the Shapley values of other players who have synergy with this player.
Specifically, if the Shapley value of player \(i\) equals \(\varphi_i < 1\), we decrease the Shapley values of other players by \((1 - \varphi_i)\) in proportions corresponding to their values \(SV_{ij}(N,v)\).

\Cref{alg:heuristic} outlines our procedure.
Note that all players will receive their lower quota before we select a player with the Shapley value smaller than 1. 
The subsequent steps mimic the behavior of the indivisible Shapley value: whenever a payoff of 1 is assigned to a player, its contribution to others is removed from the game.
In our algorithm for large games, this contribution is subtracted from the Shapley values of others players. 

To approximate the Shapley value, we use a standard method based on sampling of permutations can be used~\cite{Calvo:etal:1999}.
However, no method exists in the literature to approximate the Shapley value matrix. 
To this end, we propose such a method, depicted in \Cref{alg:matrix_approx}. 
The proof of correctness can be found in \Cref{prop:matrix_approx} in the appendix.

\begin{algorithm}[t!]
\caption{Approximation of the Shapley value matrix}
\label{alg:matrix_approx}
   
\begin{algorithmic}[1]
\REQUIRE A coalition game $(N,v)$, where $v$ is a black box, number of samples $k$
\ENSURE The Shapley value matrix $SV_{ij}(N,v)$
   \STATE $(\varphi_{ij})_{i,j \in N} \leftarrow 0$ for every $i,j \in N$
   \FOR{$k$ random permutations $\pi$}
      \FOR{$i \in N$}
          \FOR{$j \in S^{\pi}_i$ s.t. $i \le j$}
             \STATE $\varphi_{ij} \leftarrow \varphi_{ij} + \frac{1}{k} (v(S^{\pi}_i \cup \{i\}) - v(S^{\pi}_i) - v(S^{\pi}_i \setminus \{j\} \cup \{i\}) + v(S^{\pi}_i \setminus \{j\})) \cdot \sum_{t=|S^{\pi}_i|+1}^{|N|} \frac{1}{t}$
          \ENDFOR 
      \ENDFOR
   \ENDFOR
   \FOR{$i, j \in N$ s.t. $i < j$}
      \STATE $\varphi_{ji} \leftarrow \varphi_{ij}$
   \ENDFOR
   \STATE {\bfseries return} $(\varphi_{ij})_{i,j \in N}$
\end{algorithmic}
\end{algorithm}

\section{Distinguishable Objects}\label{section:positive}

So far, we have assumed that players share some number of identical objects (e.g., parliamentary seats).  
In this section, we consider a setting where objects are non-identical.  

For a set of objects $O = \{o_1,\dots,o_k\}$ of the same value, let $\mathcal{S} = (S_1, S_2, \dots, S_k)$ be a list of object owners, where each $S_i \subseteq N$ is a subset of players.  
Our goal is to find a fair allocation of objects to the players. By an \emph{allocation}, we mean a sequence $(X_1, \dots, X_n)$ such that $X_1 \cup \dots \cup X_n = O$, $X_i \cap X_j = \emptyset$ for $i \neq j$, and $o_j \in X_i$ implies $i \in S_j$.  
The last condition ensures each object is allocated to its owner.

\begin{example}\label{example:antlers}  
A group of five hunters wishes to divide trophy antlers from several hunting trips.  
The hunters listed the participants involved in acquiring consecutive trophy antlers as follows: $\{1, 2, 3\}, \{1, 2, 5\}, \{3, 4\}, \{3, 4, 5\}$. How should they divide the antlers among themselves?  
\hfill \(\lrcorner\)  
\end{example}  

To address this, we define an indivisible coalitional game $v_{\mathcal{S}}: 2^N \rightarrow \mathbb{N}$, where the value of a coalition is the number of objects owned by the coalition or any of its subsets:  
\[
v_{\mathcal{S}}(S) = |\{S_j \in \mathcal{S} : S_j \subseteq S\}| \quad \text{for every } S \subseteq N.  
\]  
This game is clearly positive: the Harsanyi dividend of any coalition is the number of its objects, which is non-negative.  
The Shapley value of this game assigns to each player an equal share of each object they own:
\begin{equation}\label{eq:shapley_matching}
SV_i(N, v_\mathcal{S}) = \sum_{S_j \in \mathcal{S}, i \in S_j} \frac{1}{|S_j|}.
\end{equation}
Since positive games are convex, we know that the indivisible Shapley value produces an indivisible payoff vector that approximates the Shapley value and belongs to the core.  
However, it is not known whether it is always possible to find an allocation of objects that matches the numbers given by the indivisible Shapley value.  
In this section, we develop an algorithm that determines the indivisible Shapley value and finds the corresponding allocation.

Our approach builds upon matching theory.  
From a list of groups $\mathcal{S}$, we create a bipartite graph $G_{\mathcal{S}} = (N, O, E)$, where the left-hand side nodes represent players $N$, the right-hand side nodes represent objects $O = \{o_1, \dots, o_k\}$, and $E$ is the set of edges connecting players to the objects they own:  
$E = \{\{i, o_j\} : i \in N, o_j \in O, i \in S_j\}$.
Our goal is to find a subset of edges $E' \subseteq E$ such that:
\begin{itemize}
\item every object is incident to exactly one edge: $|\{e \in E' : o_j \in e\}| = 1$ for every $o_j \in O$, and 
\item the number of edges incident to each player-node $i$, i.e., $|\{e \in E' : i \in e\}|$, equals the indivisible Shapley value of player $i$ in game $(N, v_{\mathcal{S}})$.
\end{itemize}
Note that the indivisible Shapley value is not given but needs to be determined by the algorithm.

As in the original method from \Cref{alg:main}, we divide the procedure into two parts.  
First, we assign to each player $i$ $\lfloor SV_i(N,v) \rfloor$ objects.  
Then, we distribute the remaining objects (potentially modifying the initial assignments).

\begin{algorithm}[t!]
   \caption{Indivisible Shapley Value for Positive Games}
   \label{alg:matching}
   
\begin{algorithmic}[1]
\REQUIRE Objects $O$, players $N$, and a list of object owners $\mathcal{S}$
\ENSURE An allocation of objects $X = (X_1, \dots, X_n)$
   \STATE create a bipartite graph $G_{\mathcal{S}} = (N,O,E)$
   \STATE compute the Shapley value in game $(N, v_\mathcal{S})$
   \STATE duplicate player-node $i$ $\lfloor SV_i(N,v_\mathcal{S}) \rfloor$ times
   \STATE find a perfect matching of player-nodes into objects (using the Hopcroft-Karp algorithm)
   \STATE $\pi \leftarrow$ permutation of players in decreasing order according to $SV_i(N,v_{\mathcal{S}}) - \lfloor SV_i(N,v_{\mathcal{S}}) \rfloor$
   \FOR{$i$ in order $\pi$}
      \IF{$SV_i(N,v_{\mathcal{S}}) > \lfloor SV_i(N,v_{\mathcal{S}}) \rfloor$}
         \STATE add one copy of player-node $i$
         \STATE find an augmenting path with the new node, if exists
	  \ENDIF
   \ENDFOR
   \STATE {\bfseries return} $(\{o_j \in O : \text{ edge } \{i,o_j\} \text{ is selected}\})_{i \in N}$
\end{algorithmic}
\end{algorithm}

\paragraph{Step 1}  
We compute the Shapley value by traversing the list of groups according to \Cref{eq:shapley_matching}.  
For each player $i \in N$, we create $\lfloor SV_i(N,v) \rfloor$ copies of this player's node.  
Next, we find a perfect matching of player-nodes into objects using the Hopcroft-Karp algorithm~\cite{Hopcroft:Karp:1973}.
Using Hall's marriage theorem, we can show that a perfect matching exists. For each coalition $S \subseteq N$, the number of objects partially owned by at least one player in $S$ equals:
\[ |\{o_j \!\in\! O : S \cap S_j \neq \emptyset\}| \geq \sum_{o_j \in O} \frac{|S \!\cap\! S_j|}{|S_j|} = \sum_{i \in S} SV_i(N,v). \]
Thus, for every coalition, the number of objects connected to player-nodes is greater than or equal to the number of these player-nodes.

\paragraph{Step 2}  
We consider players in descending order of their remainders, i.e., $SV_i(N,v_{\mathcal{S}}) - \lfloor SV_i(N, v_{\mathcal{S}}) \rfloor$.  
For each player, we add another copy of their player-node and attempt to extend the current matching by finding an augmenting path from the new node.  
If such a path exists, we extend the matching. Otherwise, this player cannot obtain an additional object without violating the core. 
%It can be shown that all objects will be assigned through this procedure.

\Cref{alg:matching} presents our method.  
The following theorem shows that it finds an allocation corresponding to the indivisible Shapley value.

\begin{theorem}
\label{theorem:matching}
Given a list of object owners $\mathcal{S}$, \Cref{alg:matching} finds an allocation of objects $(X_1, \dots, X_n)$ such that $(|X_i|)_{i \in N}$ is the indivisible Shapley value of game $(N, v_{\mathcal{S}})$ and runs in time $O(|\mathcal{S}|(|\mathcal{S}|+|N|)^{3/2} + |\mathcal{S}| |N| (|\mathcal{S}|+|N|))$.
\end{theorem}

Our algorithm runs in polynomial time with respect to the number of objects, and also to the value of the grand coalition in the corresponding game, or, after minor adaptation, to the number of non-zero Harsanyi dividends (see \Cref{proposition:positive_polynomial} in the appendix for details).
This demonstrates that it is possible to compute the indivisible Shapley value in polynomial time for positive games represented as a list of non-zero dividends. Such a representation is obtained, for example, from marginal contribution nets without negative weights or literals~\cite{Ieong:Shoham:2005}.  

\section{Case Studies}\label{section:case_studies}
In this section, we discuss three applications of our method: its use as an apportionment approach in approval elections, its application for identifying key regions in image classification tasks, and its role in ensuring a fair distribution of seats among political parties forming a coalition.

\subsection{Approval-Based Apportionment}\label{section:french_election}

Consider the problem of distributing a fixed number of parliamentary seats among political parties, where each voter $i$ indicates a subset $A_i$ of approved parties~\cite{Brill:etal:2024, Brams:etal:2019}. As discussed in the introduction, this can be viewed as an indivisible coalitional game, with parties being players and the value of a coalition $S$ defined as the fraction of voters $i$ with $A_i \subseteq S$, multiplied by the total number of seats.

Clearly, Indivisible Shapley Value can be used to distribute the seats among the parties, thus inducing an apportionment method. Below, we illustrate its behavior based on the data from the 2002 French presidential elections ~\cite{french_election,preflib}. There, we interpret each presidential candidate as a political party.\footnote{
While using approval ballots for committee elections is already well-established in the literature~\cite{lac-sko:b:approval-survey}, their application to apportionment is a recent idea. Since, no datasets directly related to approval-based apportionment are available, we use datasets from presidential elections, where, intuitively, each candidate can be seen as representing their own party.}  We treat each voter's top three preferences, indicated in the dataset, as an approval vote.
 
We compare the distribution of 100 seats computed by the Indivisible Shapley Value with the actual election outcome and the outcome produced by Proportional Approval Voting (PAV)~\cite{lspav}, a rule exhibiting particularly strong properties pertaining to proportionality.
\cref{figure:french} presents the results.

\begin{figure}[tb!]
\centering

\begin{tikzpicture}[scale=0.86]
\begin{axis}[
    xbar stacked,
    bar width=18pt, % Slightly reduced bar width
    width=0.72\textwidth, % Reduced plot width
    height=4.6cm, % Slightly reduced height
    xlabel={Share of mandates (\%)},
    ytick=data,
    yticklabels={ISV, PAV, Actual},
    xtick={0,20,40,60,80,100},
    tick label style={font=\small},
    y tick label style={xshift=-6pt}, % Added space between y-axis labels and plot
    ylabel shift=-10pt, % Adjusts ylabel position if needed
    xmin=0,
    xmax=100,
    enlarge y limits=0.2,
    legend style={
        at={(1.15,0.5)}, % Increased horizontal spacing between plot and legend
        anchor=west, % Align to left edge
        font=\small, % Increased font size
        legend columns=2, % Three columns as requested
        column sep=1em, % Increased column separation
        cells={anchor=west}, % Left-align text
        draw=none,
        legend cell align=left,
        /tikz/every even column/.append style={column sep=1em}
    },
    nodes near coords,
    every node near coord/.append style={
        font=\tiny, % Reduced size for node labels
        opacity=1
    },
]
% Muted colors with better contrast
\definecolor{gluckstein}{RGB}{195,95,90} % Brighter red
\definecolor{laguiller}{RGB}{205,115,95} % Lighter red-orange
\definecolor{besancenot}{RGB}{215,135,100} % Coral
\definecolor{hue}{RGB}{220,155,105} % Light orange
\definecolor{taubira}{RGB}{220,180,110} % Lighter amber
\definecolor{jospin}{RGB}{210,205,120} % Light yellow
\definecolor{chevenement}{RGB}{180,205,130} % Light yellow-green
\definecolor{mamere}{RGB}{145,195,135} % Light green
\definecolor{lepage}{RGB}{125,190,160} % Light teal
\definecolor{bayrou}{RGB}{110,180,180} % Turquoise
\definecolor{saintjosse}{RGB}{100,170,200} % Sky blue
\definecolor{boutin}{RGB}{90,150,195} % Lighter medium blue
\definecolor{madelin}{RGB}{85,130,185} % Lighter blue
\definecolor{chirac}{RGB}{90,115,175} % Cornflower blue
\definecolor{megre}{RGB}{100,105,170} % Periwinkle
\definecolor{lepen}{RGB}{110,95,160} % Light indigo

% Add thin white borders between segments
\addplot+[gluckstein,fill=gluckstein,draw=white,line width=0.15pt] coordinates {(1,0) (0,1) (1,2)};
\addplot+[laguiller,fill=laguiller,draw=white,line width=0.15pt] coordinates {(5,0) (2,1) (6,2)};
\addplot+[besancenot,fill=besancenot,draw=white,line width=0.15pt] coordinates {(5,0) (3,1) (4,2)};
\addplot+[hue,fill=hue,draw=white,line width=0.15pt] coordinates {(3,0) (0,1) (4,2)};
\addplot+[taubira,fill=taubira,draw=white,line width=0.15pt] coordinates {(5,0) (1,1) (2,2)};
\addplot+[jospin,fill=jospin,draw=white,line width=0.15pt] coordinates {(15,0) (30,1) (16,2)};
\addplot+[chevenement,fill=chevenement,draw=white,line width=0.15pt] coordinates {(10,0) (10,1) (5,2)};
\addplot+[mamere,fill=mamere,draw=white,line width=0.15pt] coordinates {(9,0) (8,1) (5,2)};
\addplot+[lepage,fill=lepage,draw=white,line width=0.15pt] coordinates {(5,0) (1,1) (2,2)};
\addplot+[bayrou,fill=bayrou,draw=white,line width=0.15pt] coordinates {(10,0) (14,1) (7,2)};
\addplot+[saintjosse,fill=saintjosse,draw=white,line width=0.15pt] coordinates {(2,0) (0,1) (4,2)};
\addplot+[boutin,fill=boutin,draw=white,line width=0.15pt] coordinates {(2,0) (0,1) (1,2)};
\addplot+[madelin,fill=madelin,draw=white,line width=0.15pt] coordinates {(7,0) (2,1) (4,2)};
\addplot+[chirac,fill=chirac,draw=white,line width=0.15pt] coordinates {(14,0) (25,1) (20,2)};
\addplot+[megre,fill=megre,draw=white,line width=0.15pt] coordinates {(2,0) (0,1) (2,2)};
\addplot+[lepen,fill=lepen,draw=white,line width=0.15pt] coordinates {(5,0) (4,1) (17,2)};

\legend{
    Gluckstein, Laguiller, Besancenot, Hue, Taubira, Jospin, 
    Chevènement, Mamère, Lepage, Bayrou, Saint-Josse, Boutin, 
    Madelin, Chirac, Mégret, Le Pen
}

% Title with slightly reduced size
\node[anchor=south, font=\bfseries] at (axis cs:50,3.5) {French Election Mandate Distribution};

\end{axis}
\end{tikzpicture}
\caption{Share of mandates computed by various apportionment rules for the 2002 French presidential elections experiment data.}
%Outcomes of applying different methods for appointing seats in French Election experiment. 
%From top, an actual outcome, PAV outcome and ISV outcome.
\label{figure:french}
\end{figure}
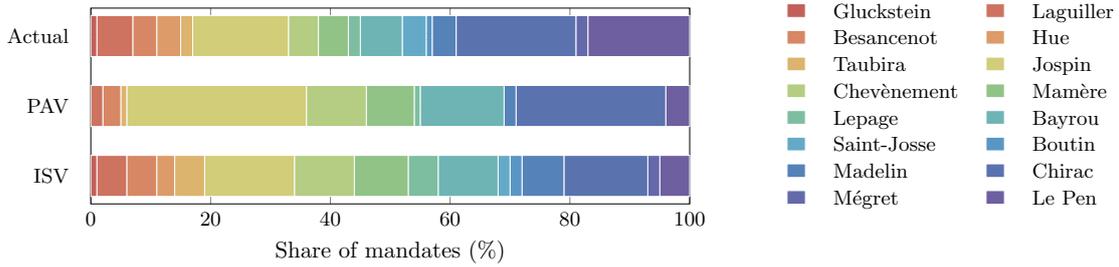

While PAV is known to select candidates fairly from the voters’ perspective, it can produce solutions that are arguably unfair to political parties. For example, if all voters approve the same two parties, PAV may allocate all seats to one party, since such a solution gives each voter the same number of representatives as evenly splitting the seats between the two. Our simulations indicate that this is likely in practice---in \Cref{figure:french} we see that smaller and extreme parties are underrepresented by PAV.

The Indivisible Shapley Value, on the other hand, provides a fairer allocation from the perspective of parties, assigning seats in a manner that more closely reflects the actual vote distribution. At the same time, by using approvals the voters can indicate their willingness for a compromise, and so ISV was able to transfer some seats from the far-right to more moderate parties, compared to the actual outcome (with each voter voting for a single candidate). Interestingly, in this instance, our method satisfies EJR+, a strong axiom of proportional representation~\cite{ejrplus}, indicating that being fair towards parties does not come at the cost of proportionally representing the voters. That said, the experiments on other datasets indicate that this property in principle must not not always be satisfied.

\subsection{Key Regions in Image Classification}\label{section:image_classification} %MC: Checklista sugeruje, że powinnismy tu jeszcze podać jakąś informację o tym jakiej maszyny używamy i ile czasu zajeło nam samplowanie.
Next, we apply \cref{alg:heuristic} for identifying key regions in an image classification task.

We begin with a step-by-step explanation of the procedure on a single example.
We use the image classification model EfficientNet B0 \cite{Tan:Le:2019} and dataset MiniImageNet that is a downscaled subset of ISVLRC ImageNet-1k~\cite{Russakovsky:etal:2015}, while retaining all 1,000 labels. We consider the image presented in \Cref{figure:bear}, which is classified as belonging to the \emph{brown bear} class.

In the first step, we use the skimage.segmentation.slic function from the scikit-image library \cite{van_der_Walt_2014} to divide the image into 247 regions using K-means clustering in the Color-(x, y, z) space~\cite{6205760}. We use the compactness parameter of 20. In order to identify key regions, we treat each region as a player and groups of regions as coalitions. The value of a coalition is defined as the probability that the model correctly classifies the image when the other regions are replaced with the average of 50 randomly selected images from the dataset.

We approximate the Shapley value matrix using \Cref{alg:matrix_approx} with 50,000 samples and estimate the Shapley value by summing the rows in the estimated matrix. We then increase the values by a constant to eliminate negative Shapley values and scale them, to ensure that the value of the grand coalition equals the number of regions we aim to select: $v(N) = 10$ in our case. Similarly, we modify the Shapley value matrix: adding a value to a player corresponds to adding this value to the diagonal of the matrix, while multiplication affects all matrix values.

\begin{figure}[bt!]
\centering
\begin{tikzpicture}[x=4cm,y=4cm]
    \node at (-1.2,0) {\includegraphics[width=0.29\textwidth]{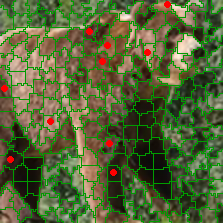}};
    \node at (0,0) {\includegraphics[width=0.29\textwidth]{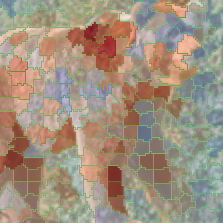}};
    \node at (1.2,0) {\includegraphics[width=0.29\textwidth]{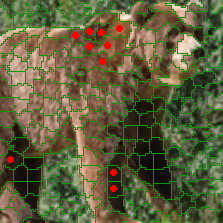}};
\end{tikzpicture}
\caption{On the left, our method with $\alpha=0.5$; in the middle, a continuous Shapley value heatmap overlay; and on the right, the baseline method which selects regions with the highest Shapley values.}
% our method with $\alpha=0.0$, which is equivalent to simply selecting regions with the highest Shapley values.}
\label{figure:bear}
\vspace{-0.3cm}
\end{figure}

% \begin{figure}[bt!]
% \centering
% \begin{tikzpicture}[x=4cm,y=4cm]
%  \node at (-1.50,0) {{\includegraphics[width=0.50\textwidth]{img/animal_images/BrownBearISV2D.png}}};

%  \def\y{0};
%  \node at (0,0.5) {{\includegraphics[width=0.22\textwidth]{img/animal_images/BrownBearSV.png}}};
%  \node at (0,-0.5) {{\includegraphics[width=0.22\textwidth]{img/animal_images/BrownBearNaive.png}}};
% \end{tikzpicture}
% \caption{Image of a bear divided into regions, with red dots indicating regions selected by our method. On the right, a comparison between two approaches: one using a continuous Shapley value heatmap overlay and a baseline selecting regions with the highest Shapley values. }
% \label{figure:bear}
% \end{figure}

\Cref{figure:bear} compares the results of our method for $\alpha=0.5$ with a Shapley value heatmap overlaid on the image and a baseline approach that simply selects regions with the highest Shapley values (which is equivalent to our method with $\alpha=0.0$). Red dots indicate the regions chosen by our method and the baseline. The heatmap colors represent feature importance: red for important regions and blue for negatively contributing ones. 

Our method identified a more diverse set of important features compared to the baseline method, which favored features with high synergy, selecting seven dots on the bear's shoulders.
The sum of Shapley values of features selected by our method was lower only by 27\% than the sum for features selected by the baseline method, while their sum of synergies was lower by 67\%.
We applied Euclidean k-means clustering~\cite{macqueen1967some} (100 runs for each $k \in [2,9]$) and selected the solution with the highest Silhouette Score~\cite{Silhouette}. 
The highest score was 0.58 for our method and 0.75 for the baseline method, both with three clusters. 
Overall, we conclude that our method trades off a small portion of the total Shapley value to achieve greater diversity, measured by lower total synergies and a lower Silhouette Score, indicating a less clustered selection of regions.

% Our method with $\alpha=0.5$ identified a more diverse set of important features compared to the baseline method ($\alpha=0.0$), which favored features with high synergy, selecting seven dots on the bear's shoulders. The sum of Shapley value for selected features was 73\% of the value for $\alpha=0.0$, while the sum of synergies was 33\% of the value. We applied Euclidean k-means clustering~\cite{macqueen1967some} (100 runs for each $k \in [2,9]$) and selected the solution with the highest Silhouette Score~\cite{Silhouette}. The highest score was 0.58 for $\alpha=0.5$ and 0.75 for $\alpha=0.0$, both with three clusters. Statistics for other images can be found in \cref{experiment:statistics}. Overall, we conclude that our method trades off a small portion of the total Shapley value to achieve greater diversity, measured by lower total synergies and a lower Silhouette Score, indicating a less clustered selection of regions.

% Sum of shapleys for alpha=0.0: 0.0993855744600296
% Sum of shapleys for alpha=0.5: 0.0724257305264473
% Ratio: 0.7287348508834839

% Sum of synergies for alpha=0.0: 0.011621389538049698
% Sum of synergies for alpha=0.5: 0.003796339500695467
% Ratio: 0.32666829228401184

% Optimal number of clusters for alpha=0.0: 3, Silhouette Score: 0.7480949537203689
% Optimal number of clusters for alpha=0.5: 3 Silhouette Score: 0.5376071740227483

\Cref{figure:zoo} presenting a comparison of both methods across different images can be found in \cref{visualization:zoo} and statistics for them are provided in \cref{experiment:statistics}. 

We ran code for our case study on a supercomputer with nodes composed of Nvidia A100 GPUs and AMD EPYC CPUs. Collecting 1000 samples for Shapley value matrix took us roughly one hour.

\subsection{Coalition formation}\label{section:polish_election}
The D'Hondt method is widely used to allocate seats in parliamentary elections. It is known to favor larger parties, a bias that becomes even more pronounced when the method is applied independently across multiple regions. Thus, smaller parties can increase their representation by forming a coalition and running on a single list, receiving more votes than the combined total of their separate lists. 

The Indivisible Shapley Value offers a principled way to fairly distribute the seats gained by forming such a coalition.
Specifically, we model a cooperative game in which the coalition parties are the players, and the value of any subcoalition is defined by the number of seats it would secure.

As an example, we investigate the 2023 Polish parliamentary elections, in which three opposition parties -- Civic Coalition, Third Way, and New Left -- considered forming a coalition but ultimately ran separately. Using data from the Polish National Electoral Commission, we simulated various coalition scenarios by aggregating the votes received by each subcoalition across all regions.

% \begin{table}[H]
%   \caption{2023 Polish Elections}
%   \label{sample-table}
%   \centering
%   \begin{tabular}{llll}
%     \toprule
%     Method & Civic Coalition & Third Way & New Left    \\
%     \midrule
%     Real results & 157 & 65 & 26 \\
%     % Shapley Value & 160.5 & 72 & 32.5 \\
%     Indivisible Shapley Value & 160 & 72 & 33 \\
%     \bottomrule
%   \end{tabular}
% \end{table}%TODO: napisać wyniki inline zamiast w tabelce

We found that all three parties would gain additional seats by forming a coalition. Without forming a coalition, Civic Coalition, Third Way and New Left received 157, 65 and 26 seats. In coalition after applying Indivisible Shapley Value they would receive 160, 72 and 33 seats, respectively.
%In our procedure, the Shapley Value is rounded up for the New Left and down for the Civic Coalition. 
%OS: może tego jednak nie pisać aby nie zostawiać recenzenta na koniec pracy z informacją, że zwykle gry nie są convex?
Notably, both the Shapley Value and its indivisible counterpart lie outside the core. We found that, in both the 2023 and 2019 Polish parliamentary elections, only games formed by coalitions of two parties are convex. This suggests that the most stable coalition would be formed by the New Left and Third Way.

\vspace{-0.15cm}
\section{Conclusions}
\vspace{-0.15cm}

We proposed a novel model of indivisible coalitional games and defined an indivisible Shapley value, an extension of the classic Shapley value, assigning each player an integer. 
We proved it belongs to the core in convex integer games and proposed an algorithm that computes a corresponding allocation for non-identical objects. 
Morevoer, we introduced a method that mimics the indivisible Shapley value for large games and tested our approach on three case studies.

There is a wide array of potential future directions. 
Solution concepts other than the Shapley value can be analyzed for our model. We find randomized algorithms allowing to achieve fairness in expectation a particularly interesting direction.
Also, properties of our method can be studied in specific applications.
Furthermore, it would be interesting to test concepts from the social choice literature in our setting.

% \section{Conclusions}

% We proposed a model of indivisible coalitional games and defined an indivisible Shapley value, an extension of the classic Shapley value, assigning each player an integer. 
% We proved it belongs to the core in convex integer games and proposed an algorithm that computes a corresponding allocation for non-identical objects. 
% Morevoer, we introduced a method that mimics the indivisible Shapley value for large games and tested our approach on image classification task.

% There is a wide array of potential future directions. 
% Solution concepts other than the Shapley value can be analyzed for our model.
% Also, properties of our method can be studied in specific applications.
% Furthermore, it would be interesting to test concepts from the social choice literature on our setting.

\section*{Acknowledgement}
Oksar Skibski was supported by the Polish National Science Centre Grant No.2023/50/E/ST6/00665.
Miko{\l}aj Czarnecki, Micha{\l} Korniak and Piotr Skowron were supported by European Union (ERC, PRO-DEMOCRATIC,
101076570). 
Views and opinions expressed are however those of the author(s) only and do not necessarily reflect those of the European Union or the European Research Council.
Neither the European Union nor the granting authority can be held responsible for them. 
We acknowledge Polish high-performance computing infrastructure PLGrid (HPC Center: ACK Cyfronet AGH) for providing computer facilities and support within computational grant no. PLG/2025/018243
\vspace{-0.3cm}\begin{center}\includegraphics[width=5cm]{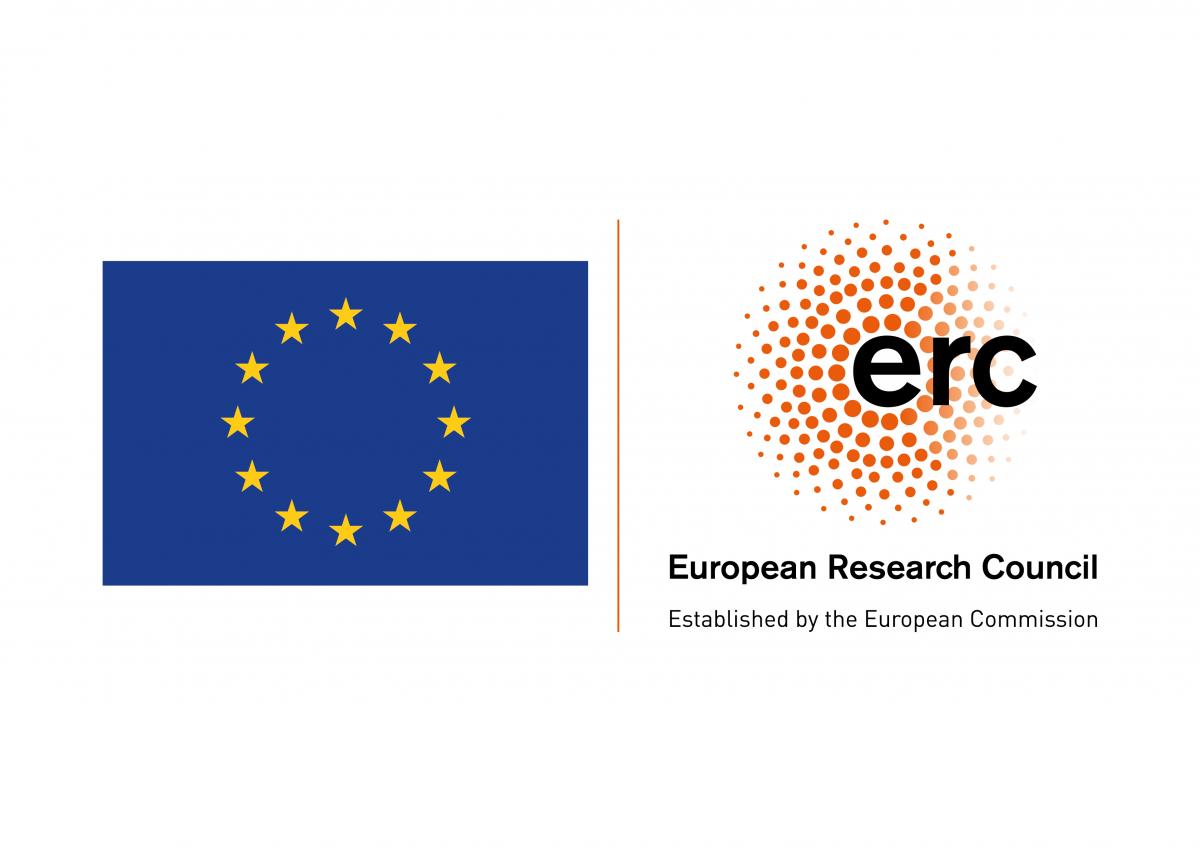}\end{center}

\bibliography{bibliography}
\bibliographystyle{plainnat}

\clearpage

\appendix
\section*{Appendix}

Throughout the appendix, for $x \in \mathbb{R}^N$ and $S \subseteq N$, we will write $x(S)$ to denote $\sum_{i \in S} x_i$.

\section{Supplementary Materials for \Cref{section:convex}}
%\section{Proofs for \Cref{section:convex,section:general_games}}

This section contains missing proofs of the results from Section 4.

\begin{proof}[Proof of \Cref{lemma:reduced_game}]
First, let us consider size-boundness. 
For $S = N \setminus \{i\}$ we have $v(N \setminus \{i\}) = v(N) - c < |N|-1$ from the assumption that $c \ge 1$.
For any other non-empty coalition $S \subsetneq N \setminus \{i\}$, from the definition of a reduced game, we have $\Psi^{i \rightarrow c}_v(S) = \max\{v(S), v(S \cup \{i\})-c\} < \max \{ |S|, (|S|+1)-c\} = |S|$.
Hence, $\Psi^{i \rightarrow c}_v$ is also size-bounded.

Let us focus on convexity.
We need to prove that for an arbitrary $j \in N \setminus \{i\}$, $S \subseteq T \subseteq N \setminus \{i,j\}$:
\begin{equation}
\label{eq:proof_reduced_convex}
\Psi^{i \rightarrow c}_{v}(S\cup\{j\})-\Psi^{i \rightarrow c}_{v}(S) \leq \Psi^{i \rightarrow c}_{v}(T\cup \{j\})-\Psi^{i \rightarrow c}_{v}(T).
\end{equation}
Notice that for every $R \subseteq N \setminus \{i\}$ we have:
\[ \Psi^{i \rightarrow c}_{v}(R) = \begin{cases}
v(R \cup \{i\})-c & \text{if }v(R \cup \{i\})-v(R) > c,\\
v(R) & \text{otherwise.}
\end{cases} \]
For $R \not \in \{\emptyset, N \setminus \{i\}\}$ this is clear from the definition of the reduced game.
For $R = N \setminus \{i\}$, we have $\Psi^{i \rightarrow c}_v(N) = v(N) - c$ and we assumed $v(N) - v(N \setminus\{i\}) \ge c$, hence the equality holds.
For $R = \emptyset$, we have $\Psi^{i \rightarrow c}_v(N) = 0 = v(\emptyset)$ and $v(\{i\}) - v(\emptyset) \le c$, also proving the equality.

From convexity of $v$, we know that if for some $R \subseteq N\setminus\{i\}$, $v(R\cup\{i\})-v(R) > c$, then for any $R': R\subseteq R'\subseteq N\setminus\{i\}$ we have $v(R'\cup\{i\})-v(R')>c$.
This implies that if $\Psi^{i \rightarrow c}_{v}(R) \neq v(R)$, then $\Psi^{i \rightarrow c}_{v}(R')\neq v(R')$. Therefore, condition $\Psi^{i \rightarrow c}_{v}(R) \neq v(R)$ holds for one of the following groups of coalitions:
\begin{itemize}
\item[1.]$T\cup\{j\}$
\item[2.]$T\cup\{j\}$, $T$
\item[3.]$T\cup \{j\}$, $S\cup\{j\}$
\item[4.]$T\cup\{j\}$, $T$, $S\cup\{j\}$
\item[5.]$T\cup\{j\}$, $T$, $S\cup\{j\}$, $S$
\end{itemize}
We consider each case separately:

\begin{itemize}
\item \emph{Case 1:} we have $\Psi^{i \rightarrow c}_{v} (T\cup\{j\})-\Psi^{i \rightarrow c}_{v} (T) > v(T \cup \{j\}) - v(T) \ge v(S \cup \{j\}) - v(S) = \Psi^{i \rightarrow c}_{v}(S\cup\{j\})-\Psi^{i \rightarrow c}_{v}(S)$.

\item \emph{Case 2:} we have $\Psi^{i \rightarrow c}_{v} (T\cup\{j\})-\Psi^{i \rightarrow c}_{v} (T)=v(T\cup\{i,j\})-v(T\cup\{i\})\geq v(S\cup\{j\})-v(S)=\Psi^{i \rightarrow c}_{v}(S\cup\{j\})-\Psi^{i \rightarrow c}_{v}(S)$.

\item \emph{Case 3:} we have $\Psi^{i \rightarrow c}_{v}(T\cup\{j\}) - \Psi^{i \rightarrow c}_{v}(S\cup\{j\}) = v(T \cup \{i,j\}) - v(S \cup \{i,j\}) \ge v(T) - v(S) = \Psi^{i \rightarrow c}_{v}(T)-\Psi^{i \rightarrow c}_{v}(S)$ which is equivalent to \Cref{eq:proof_reduced_convex}.

\item \emph{Case 4:} we have $\Psi^{i \rightarrow c}_{v} (T\cup\{j\})-\Psi^{i \rightarrow c}_{v} (T)=v(T\cup\{i,j\})-v(T\cup\{i\})\geq v(S\cup\{i,j\})-v(S \cup \{i\}) \ge v(S \cup \{j\}) + c - v(S \cup \{i\}) \ge v(S \cup \{j\}) - v(S)$.
Here, we know that $c \ge v(S \cup \{i\}) - v(S)$, because $\Psi^{i \rightarrow c}_{v}(S) = v(S)$.

\item \emph{Case 5:} we have $\Psi^{i \rightarrow c}_{v}(T\cup\{j\})-\Psi^{i \rightarrow c}_{v}(T) =  v(T\cup\{i,j\})-v(T\cup\{i\}) \ge v(S\cup\{i,j\})-v(S\cup\{i\}) = \Psi^{i \rightarrow c}_{v}(S\cup\{j\})-\Psi^{i \rightarrow c}_{v}(S)$.
\end{itemize}
This concludes the proof.
\end{proof}

Before we move to the proof of \Cref{lemma:in-core-iff-marg-1}, we state an additional lemma that we will use in our proofs.

\begin{lemma}
\label{lemma:core_from_reduced_game}
For every game $(N,v)$, player $i \in N$ and $c \in \mathbb{R}_{\ge 0}$, if $y \in Core(N \setminus \{i\}, \Psi^{i \rightarrow c}_{v})$, then for $x \in \mathbb{R}^N$, $x_i = c$ and $x_j = y_j$ for every $j \in N \setminus \{i\}$ it holds $x \in Core(N,v)$.
\end{lemma}
\begin{proof}[Proof of \Cref{lemma:core_from_reduced_game}]
Take any coalition $S \subseteq N \setminus \{i\}$.
Since $y$ belongs to the core of game $(N \setminus \{i\}, \Psi^{i \rightarrow c}_{v})$, we have:
\[ y(S) \ge \Psi^{i \rightarrow c}_{v}(S). \]
Moreover, from the definition of the reduced game we know that:
\[ \Psi^{i \rightarrow c}_{v}(S) \ge v(S \cup \{i\}) - c \mbox{ and }\Psi^{i \rightarrow c}_{v}(S) \ge v(S). \]
This gives $x(S) = y(S) \ge v(S)$ and 
\[ x(S \cup \{i\}) = y(S) + c \ge \Psi^{i \rightarrow c}_{v}(S) + c \ge v(S \cup \{i\}) \]
which concludes the proof that $x$ is in the core of game $(N,v)$.
\end{proof}

\begin{proof}[Proof of \Cref{lemma:in-core-iff-marg-1}]
The proof proceeds by induction on the number of players.
If $|N|$ = 1, then the statement trivially holds: $v(N) = 0$ and the core consists of one payoff vector $x = (0)$.
In what follows, assume that the statement holds for every game with less than $n$ players and fix a convex, size-bounded integer game $(N,v)$ with $|N|=n$, and a player $i \in N$.
We have two cases depending on the value of $v(N) - v(N \setminus \{i\})$.

If $v(N) - v(N \setminus \{i\}) \ge 1$, then let us consider a reduced game $(N \setminus \{i\}, \Psi^{i \rightarrow 1}_{v})$. 
Hence, from \Cref{lemma:reduced_game}, we know that this game is size-bounded and convex (note that the assumption $v(\{i\}) < 1$ follows from size-boundness).
This fact combined with the inductive assumption implies that there exists a 0-1 payoff vector $y$ which is in the core of this game. 
Hence, from \Cref{lemma:core_from_reduced_game}, a payoff vector defined as follows: $x_i = 1$ and $x_j = y_j$ for every $j \in N \setminus \{i\}$ is in the core of $(N,v)$. 

If $v(N) - v(N \setminus \{i\}) < 1$, then $x(N \setminus \{i\}) \ge v(N \setminus \{i\}) = v(N) = x(N)$ which implies $x_i=0$. It remains to argue the the core is non-empty. This however, follows from the first case: because $v(N) > 0$ for a convex game there must be a player $i$ with a positive contribution to the grand coalition and we know that there exists a 0-1 payoff vector in the core of $(N,v)$.
\end{proof}

\subsection{Proof of \Cref{theorem:convex}}
%\subsection{Proof of \Cref{theorem:convex}}

We start by proving that \Cref{alg:main} always returns an outcome that satisfies Efficiency, Lower Quota and Upper Quota.	

\begin{proposition}
\label{proposition:general_lower_upper_efficiency}
The indivisible Shapley value (\Cref{alg:main}) satisfies Efficiency, Lower Quota and Upper Quota for every game $(N,v)$.
\end{proposition}
\begin{proof}

In line 2, we give to every player the floor of its Shapley value, i.e. $\varphi_i \leftarrow \lfloor SV_i(N, v) \rfloor$ for each $i \in N$. 
Later on, we do not decrease payoffs, so the Lower Quota is satisfied. 

Furthermore, we notice that in line 2 the game $(N, u)$ is defined as a difference between the game $(N, v)$ and the game $(N, w)$, where $w(S) = \sum_{i \in S} \lfloor SV(N,v) \rfloor $ for each $S \subseteq N$. 
Game $(N,w)$ is additive, so $SV_i(N,w) = w(\{i\}) = \lfloor SV_i(N, v) \rfloor$ for each $i \in N$. 
From Additivity we know that for every player $i \in N$ it holds $SV_i(N, u) = SV_i(N, v - w) = SV_i(N, v) - SV_i(N, w) = SV_i(N, v) - \lfloor SV_i(N,v) \rfloor$.

Now, if $SV_i(N,v) = \lfloor SV_i(N, v) \rfloor$ holds for some player, then this player is removed from the game in lines 5--7. 
Hence, Upper Quota is satisfied.

Assume otherwise that $SV_i(N,v) > \lfloor SV_i(N, v) \rfloor$.
Clearly, it holds $\lceil SV_i(N,v) \rceil = \lfloor SV_i(N, v) \rfloor+1$.
In the while loop (lines 9--13) the player can increase its payoff by $1$ and it is deleted right after, so Upper Quota is again satisfied.

It remains to prove Efficiency.
Clearly, after lines 1--4, the value of the grand coalition decreased by $\sum_{i \in N} \lfloor SV_i(N,v) \rfloor$ which equals the payoffs of players assigned so far.
In lines 5--8, based on the definition of the reduced game, the value of the grand coalition did not change.
Also, from Efficiency of the Shapley value, we know that $|M| > u(M) \ge 0$.
    
It suffices to show that in each iteration of the while loop, there is a player that can receive one item. 
Specifically, we need to show that in each iteration there is a player $i \in M$ with $SV_i(M,u) > 0$ or $u(M) \ge u(M \setminus \{i\}) + 1$. 
This follows from Efficiency of the Shapley value: there must be a player $i$ with $SV_i(N,v) > 0$.
Therefore, we distribute exactly $u(M)$ items in the while loop, so altogether $\sum_{i \in N} \varphi_i = u(M) + \sum_{i \in N} \lfloor SV(N,v) \rfloor = v(N)$, thus Efficiency is satisfied.
\end{proof}

\begin{lemma}\label{lemma:alg_before_loop}
If the input game $(N,v)$ in \Cref{alg:main} is a convex integer game, then the game $(M,u)$ obtained after lines 1--7 is a convex integer size-bounded game.
\end{lemma}
\begin{proof}
Since the input game $(N,v)$ is a convex integer game, the game $(M,u)$ in line 4 is also convex.
This is because the marginal contribution of each player to all coalitions decreases by the same amount, so it is still larger for larger coalitions: for every $i \in N$ and $S \subseteq T \subseteq N \setminus \{i\}$:
\begin{multline*}
u(S \cup \{i\}) - u(S) = \lfloor SV_i(N,v) \rfloor + v(S \cup \{i\}) - v(S) \le  \\
\le \lfloor SV_i(N,v) \rfloor + v(T \cup \{i\}) - v(T) = u(T \cup \{i\}) - u(T).
\end{multline*}

Let us now consider removing players with the zero Shapley value in game $(M,u)$ in lines 5--7.
Let $R$ be the set of such players and assume $(M \setminus R', u')$ is the game obtained after removing players $R' \subseteq R$.
Based on the definition of 0-reduced game we have:
\[ u'(S) = \max_{T \subseteq R'} u(S \cup T) \text{ for every }S \subsetneq M \setminus R' \]
From this observation, we get that:
\begin{itemize}
\item The grand coalition has the highest value in the game: for every coalition $S \subseteq M \setminus R'$ it holds $u'(M \setminus R') \ge u'(S)$. Since the core of $(M,u)$ is non-empty (the game is convex), we know that the grand coalition has the highest value in the game. Then, in the 0-reduced game, its value does not change and values of other coalitions cannot exceed the highest value in the game.
\item Players from $R \setminus R'$ has non-positive singleton values: for each player $i \in R \setminus R'$ it holds $u'(\{i\}) \le 0$. To see that observe that the Shapley value is in the core which implies $u(T) \le 0$ for every $T \subseteq R$.
Hence, $u'(\{i\}) = \max_{T \subseteq R'} u(T \cup \{i\}) \le 0$.
\item Game $(M \setminus R', u')$ is convex. The two above observations imply that $u'(M') - u(M' \setminus \{i\}) \ge 0 \ge u'(\{i\})$ holds for every $i \in R \setminus R'$. Hence, from \Cref{lemma:reduced_game} the 0-reduced game is also convex.

\item Game $(M \setminus R', u')$ is size-bounded. Assume it is not and $u'(S) \ge |S|$. Then, it must exists $T \subseteq R'$ such that $u(S \cup T) \ge |S|$. This, however, leads to the contradiction with the fact that the Shapley value is in the core: the Shapley value of players from $S \cup T$ equals $\sum_{j \in S \cup T} SV_i(M,u) = \sum_{j \in S} SV_i(M,u) < |S|$.

\end{itemize}

We get that after removing each player $i \in R$ the game remains convex and size-bounded. 
Hence, game after lines 5--7 is convex and size-bounded.
\end{proof}

\begin{lemma}\label{lemma:convex_lower_upper_efficiency}
If a convex integer game is the input of \Cref{alg:main}, then the code marked in blue does not affect its operation.
\end{lemma}
\begin{proof}
From \Cref{lemma:alg_before_loop} we know that the game $(M,u)$ obtained after lines 1--7 is a convex integer size-bounded game.
Hence, the command $u(S) \leftarrow \lfloor u(S) \rfloor$ in line 8 does not modify any values.

Now, consider the first iteration of the loop in lines 9--13.
Since the game is convex, positive Shapley value ($SV_i(M,u) > 0$) implies the player has a positive marginal contribution to some coalition which from convexity implies the marginal contribution to the grand coalition is also positive ($u(M) > u(M \setminus \{i\})$).
This, in integer games, so it is larger than or equal to $1$.
Hence, the alternative in blue does not change the condition of choosing a player.
After choosing the first player, the game is replaced by the 1-reduced game $\Psi^{i \rightarrow 1}_i$. 
From \Cref{lemma:reduced_game} we know that this game is also convex and size-bounded. 
Hence, our analysis applies to all iterations, not only the first one.
This concludes the proof.
\end{proof}

\begin{lemma}\label{lemma:isv_in_core}
For convex integer games, the indivisible Shapley value belongs to the core.
\end{lemma}
\begin{proof}
Consider the game $(N,u)$ defined in line 4 of \Cref{alg:main} as follows (note that $M=N$ at this point):
\[ u(S) = v(S) - \sum_{i \in S} \lfloor SV_i(N,v) \rfloor. \]
Since game $v$ is convex, clearly $u$ is also convex.
Moreover, if we add $(\lfloor SV_i(N,v) \rfloor)_{i \in N}$ to the payoff vector which belongs to the core of game $(N,u)$ we will obtain a payoff vector from the core of game $(N,v)$.
Hence, it remains to prove that the payoff vector added in lines 5--13 belongs to the core of game $(N,u)$.

From \ref{lemma:alg_before_loop} we know that the game $(M,u)$ obtained after lines 1--7 is a convex integer size-bounded game.
Moreover, we also know from \Cref{lemma:core_from_reduced_game} that any payoff vector that belongs to the core of $(M,u)$ extended by assigning payoff $0$ to all players from $N \setminus M$ is in the core of game $(N,u)$.

Finally, since game $(M, u)$ is convex and size-bounded, and we know that in each iteration of the while loop, player $i$ satisfies condition $u(M) - u(M \setminus \{i\}) \ge 1$, we can use Lemma~\ref{lemma:in-core-iff-marg-1} to show recursively that for any 0-1 payoff vector $y \in Core(M \setminus \{i\}, \Psi^{i \rightarrow 1}_{v})$, a payoff vector $x$ defined as $x_i = 1$ and $x_j = y_j$ for every $j \in M \setminus \{i\}$ is in the core of $(M,u)$. Therefore, a payoff vector obtained in the while loop is in the core of $(M,u)$ and thereby, in the core of $(N,u)$. Thus, indivisible Shapley value belongs to the core of $(N,v)$.
\end{proof}

Let us introduce an additional definition.
We say that a payoff vector $x$ is \emph{lexicographically larger with respect to a permutation $\pi$} if for the first player $i$ from $\pi$ for which payoff vectors differ it holds $x_i > y_i$.
Consequently, payoff vector $x$ is \emph{lexicographically maximal} (with respect to a permutation $\pi$) if there is no payoff vector $y$ such that $y$ is lexicographically larger.

\begin{lemma}
\label{lemma:lex}
Let $\pi$ be a permutation of players according to $SV_i(N,v) - \lfloor SV_i(N,v) \rfloor$ decreasingly. The indivisible Shapley value (\Cref{alg:main}) is the lexicographically maximal payoff with respect to permutation $\pi$ among all payoff vectors that satisfy Efficiency, Lower Quota and Upper Quota, and belong to the core.
\end{lemma}
\begin{proof}
Assume that the indivisible Shapley value $x$ is not the lexicographically maximal.
Let $y$ be a payoff vector lexicographically larger and take the first player $i$ according to the permutation $\pi$ on which they differ, i.e., $x_i < y_i$.
Since both payoff vectors satisfy Lower Quota and Upper Quota is must hold that $x_i = \lfloor SV_i(N,v) \rfloor < SV_i(N,v)$ which means that at some point in line 10 of \Cref{alg:main} player $i$ was not selected and some player, further in the permutation, was selected.
This means that at this moment it was true that $u(M) = u(M \setminus \{i\})$.
But \Cref{lemma:in-core-iff-marg-1} implies that the core of game $(M,u)$ does not contain any payoff vector in which player $i$ has $1$ which leads to the contradiction with the claim that $y$ belongs to the core.
\end{proof}

\begin{lemma}
\label{lemma:convexintersection}
For every convex integer game $(N,v)$, payoff vector $x \in Core(N,v)$, and coalitions $S, T \subseteq N$ such that $x(S) = v(S)$ and $x(T) = v(T)$ it holds $x(S \cup T) = v(S \cup T)$ and $x(S \cap T) = v(S \cap T)$.
\end{lemma}
\begin{proof}
We know that 
\[ x(S) + x(T) = x(S \cup T) + x(S \cap T). \] 
On the other hand, from convexity of $(N,v)$ it holds that 
\[ v(S) + v(T) \le v(S \cup T) + v(S \cap T). \] 
From the assumption that $x(S) = v(S)$ and $x(T) = v(T)$ we get that:
\[ x(S \cup T) + x(S \cap T) \le v(S \cup T) + v(S \cap T). \] 
Since $x$ belongs to the core, we also know that $x(S \cup T) \ge v(S \cup T)$ and $x(S \cap T) \ge v(S \cap T)$, which combined with the above inequality implies that $x(S \cup T) = v(S \cup T)$ and $x(S \cap T) = v(S \cap T)$.
\end{proof}

\begin{lemma}
\label{lemma:swapping1}
For a convex integer game $(N,v)$, two payoff vectors $x,y \in \{0,1\}^N \cap Core(N,v)$ and player $i \in N$, if $x_i = 1$ and $y_i = 0$, then there exists a player $j$ such that $x_j = 0$ and $y_j = 1$ and the payoff vector obtained by changing $x_i$ to $0$ and $x_j$ to $1$ belongs to the core of $(N,v)$.
\end{lemma}

\begin{proof}
Let us consider two payoff vectors $x,y$ as described in the lemma statement. 
Let $i \in N$ by any player so that $x_i = 1$ and $y_i = 0$. 
By contradiction, we assume that there is no player $j$ such that $x_j=0$, $y_j=1$ and at the same time payoff vector obtained by changing $x_i$ to $0$ and $x_j$ to $1$ belongs to the core of $(N,v)$.

Let $T$ denote a set of players $k \in N$ such that $x_k = 0$ and $y_k = 1$. 
Let $S_1, S_2, ..., S_m$ be a list of all coalitions that contain player $i$ such that $x(S_{\ell}) = v(S_{\ell})$. 
These are the only coalitions for which decreasing the payoff of $x_i$ by one may violate their core condition $x(S) \geq v(S)$.
Note that this list contains at least the grand coalition $N$.
Let $S_{\cap} = \bigcap_{{\ell}=1}^{m} S_{\ell}$ be the intersections of all these coalitions. 
From Lemma~\ref{lemma:convexintersection} we know that $v(S_{\cap}) = x(S_{\cap})$. 
Now, if a player $j$ from $T$ belongs to $S_{\cap}$, then the payoff vector obtained by changing $x_i$ to $0$ and $x_j$ to $1$ would also be in the core which is a contradiction.
On the other hand, if $T \cap S_{\cap} = \emptyset$, then together with fact that $i \in S_{\cap}$, we notice that $y(S_{\cap}) < x(S_{\cap}) = v(S_{\cap})$, thus $y$ is not in the core which is also a contradiction.
\end{proof}

\begin{proof}[Proof of \Cref{theorem:convex}]
It is clear from \Cref{alg:main} that the indivisible Shapley value is an indivisible payoff vector.
We already argued in \Cref{proposition:general_lower_upper_efficiency,lemma:convex_lower_upper_efficiency} that it satisfies Efficiency, Lower Quota and Upper Quota.
Also, we proved in \Cref{lemma:isv_in_core} that the indivisible Shapley value belongs to the core.
Thus, it remains to prove that the indivisible Shapley value has a minimal distance to the Shapley value out of all payoff vectors satisfying the above conditions.
More precisely, we will prove that for any $p$ the indivisible Shapley value $x$ satisfies
\[ x = \arg \min_{y} \{ ||(SV_i(N,v))_{i \in N} - y||_p : y \mbox{ satisf. (a)-(c)}\} \] 
Since $y$ satisfies Lower Quota and Upper Quota, it is enough to focus on games in which $SV_i(N,v) \in [0,1)$ holds for every $i \in N$ and $y \in \{0,1\}^{N}$.

Let $\pi$ be a permutation of players according to $SV_i(N,v) - \lfloor SV_i(N,v) \rfloor$ decreasingly.
By contradiction, assume the indivisible Shapley value $x$ does not minimize the $L^p$-distance to the Shapley Value.
Let $y$ be a payoff vector that satisfies all conditions and minimizes this distance, and among such payoff vectors it is lexicographically maximal (with respect to $\pi$).
Let $i$ be the first player, according to $\pi$, on which vectors $x$ and $y$ differ.
Since we proved in \Cref{lemma:lex} that $x$ is the lexicographically maximal we have $x_i=1$ and $y_i=0$.
Furthermore, we know from \Cref{lemma:swapping1} that there exists a position $j$ after $i$ ($\pi(i) < \pi(j)$), such that $x_j = 0$ and $y_j = 1$ and the payoff vector $y'$ obtained from $y$ by changing $y_i$ to $1$ and $y_j$ to $0$ is in the core.

We have: 
\begin{multline*} 
||SV(N,v) - y'||_p = \left(\sum_{{\ell} \in N} (|SV_{\ell}(N,v) - y'_{\ell}|^p)\right)^{\frac{1}{p}} = \\
= \left(\sum_{{\ell} \in N\setminus \{i,j\}} \left(|SV_{\ell}(N,v) - y'_{\ell}|^p) + |SV_i(N,v) - y'_i|^p + |SV_j(N,v) - y'_j|^p \right)\right)^{\frac{1}{p}}.
\end{multline*}
Likewise: 
\[ ||SV(N,v) - y||_p =
\left(\sum_{{\ell} \in N\setminus \{i,j\}} \left(|SV_{\ell}(N,v) - y_{\ell}|^p) + |SV_i(N,v) - y_i|^p + |SV_j(N,v) - y_j|^p \right)\right)^{\frac{1}{p}}. \] 
Since, $\pi$ is ordered by $SV_{\ell}(N,v) - \lfloor SV_{\ell}(N,v) \rfloor$ decreasingly and $\lfloor SV_{\ell}(N,v) \rfloor = 0$ for every player ${\ell} \in N$, because $SV_{\ell}(N,v) \in [0,1)$, we have that $SV_i(N,v) \ge SV_j(N,v)$. Therefore, $|SV_i(N,v) - y'_i|^p \le |SV_j(N,v) - y_j|^p$, because $y'_i = 1$ and $y_j = 1$, and $|SV_j(N,v) - y'_j|^p \le |SV_i(N,v) - y_i|^p$, because $y'_j = 0$ and $y_i = 0$. 
Thus, $||SV(N,v) - y'||_p \le ||SV(N,v) - y||_p$ and we have a contradiction, because $y'$ is lexicographically larger than $y$ and $||SV(N,v) - y'||_p \le ||SV(N,v) - y||_p$.
\end{proof}

%%%%%%%%%%%%%%%%%%%%%%%%%%%%%%%%%%%%%%%%%%%%%%%%%%%%%%%%%%%%%%%%%%%
\section{Supplementary Materials for \Cref{section:general_games}}\label{appendix:general_games}
%%%%%%%%%%%%%%%%%%%%%%%%%%%%%%%%%%%%%%%%%%%%%%%%%%%%%%%%%%%%%%%%%%%

This section contains the full description of the extension of the Indivisible Shapley Value to general games.

\subsection{Indivisible Shapley Value for general games}

Let us discuss how to generalize our method for all games.
We begin with general convex fractional games.  
As we have argued, the core of such games is non-empty and contains the Shapley value.  
However, in contrast to convex integer games, it might happen that no indivisible payoff vector lies in the core.  
In fact, indivisible payoff vectors might be arbitrarily far from the core.

\begin{example}
\emph{(Fractional-Half-Game)}  
Consider an integer indivisible game with $N = \{1, \dots, n\}$, where $n$ is even, and $v(S) = |S|/2$.  
We have $v(N) = n/2$, and clearly $x = [0.5, \dots, 0.5]$ is the only payoff vector in the core.  
However, in every indivisible payoff vector, there exist at least $n/2$ players with a payoff of $0$.  
Thus, for the coalition $S$ of these players, we have $\sum_{i \in S} x_i \!=\! 0 \!<\! n/4 \!=\! v(S)$.  
\hfill $\lrcorner$  
\end{example}

To adapt our method to general fractional games, we introduce two changes (marked in blue in \Cref{alg:main}).  
First, after assigning the floor of the Shapley value to all players and removing players with zero Shapley value from the resulting game, we round down all the values in the game (line 8).  
We round values down because we want to assess coalitions based on what they are entitled to.  
In indivisible games, payoffs are whole numbers, so assuming that a coalition with the value $x$ is entitled to $x$ is equivalent to assuming it is entitled to $\lceil x \rceil$, which might unfairly prioritize coalitions with very small values.  
This issue arises, in particular, in games with a large number of players and small coalition values, such as the feature-games derived from machine learning models.  

Second, we modify the condition used to select the player who receives an additional $1$ in the final loop (line 10).  
In convex games, we selected a player with a marginal contribution to the grand coalition larger than or equal to $1$.  
However, in general games, such a player might not exist.  
Thus, we select a player whose marginal contribution is at least $1$ or who has a positive Shapley value in the current game.  
Note that since Shapley values always sum to the value of the grand coalition, such a player always exist.  

We note that both modifications do not change how the procedure works in the case of convex integer games.  
Clearly, rounding down the values does not affect integer values.  
Moreover, if a player's marginal contribution to the grand coalition is not positive, then by convexity, all marginal contributions are non-positive.  
Hence, the Shapley value does not exceed $0$.  
This means no player can satisfy $SV_i(M,u)>0$ without also satisfying $v(M) \ge v(M \setminus \{i\}) + 1$.  In \Cref{proposition:general_lower_upper_efficiency} we proved that \Cref{alg:main} in general games also satisfies Efficiency, Lower Quota and Upper Quota.

\subsection{ISV-Approach for Large Games}
The indivisible Shapley value operates under the assumption that the full coalitional game is given as input. However, since the size of the game grows exponentially with respect to the number of players, this assumption is not realistic. In this section, we present an algorithm that simulate the behavior of the indivisible Shapley value: approximates the Shapley value and then operates in polynomial time to determine the final outcome.

\begin{algorithm}[t!]
   \caption{ISV-Approach for Large Games}
   \label{alg:heuristic}
   
\begin{algorithmic}[1]
\REQUIRE An indivisible game \((N,v)\), where \(v\) is a black box
\ENSURE A payoff vector \(x \in \mathbb{Z}^N\)
   \STATE \(x_i \leftarrow 0\) for every \(i \in N\)
   \STATE \((\varphi_i)_{i \in N} \leftarrow\) approx. of the Shapley values
   \STATE \((\varphi_{ij})_{i,j \in N} \leftarrow\) approx. of the Shapley value matrix
   \WHILE{\(v(N) > \sum_{i \in N} x_i\)}
      \STATE \(i \leftarrow\) player with the highest value \(\varphi_i\)
      \STATE \(x_i \leftarrow x_i + 1\)
      \IF{\(\varphi_i > 1\)}
         \STATE \(\varphi_i \leftarrow \varphi_i - 1\)
      \ELSE
      	\STATE \(\varphi_j \leftarrow \varphi_j - (1 - \varphi_i) \left(\frac{\alpha \cdot \varphi_{ij}}{{\sum_{k \in N \setminus \{i\}} \varphi_{ik}}} + \frac{1-\alpha}{n-1}\right)\) for every \(j \in N \setminus \{i\}\)
      	\STATE \(\varphi_i \leftarrow 0\)
      \ENDIF
   \ENDWHILE
   \STATE {\bfseries return} \(x\)
\end{algorithmic}
\end{algorithm}

We utilize the notion of the Shapley value matrix $(SV_{ij}(N,v))_{i,j \in N}$~\cite{Hausken:Mohr:2001}. 
In this matrix, the Shapley value of a player is decomposed into a vector $(SV_{i1}(N,v), \dots, SV_{in}(N,v))$, where each element represents the synergy between two players: the contribution of player \(j\) to the value of player \(i\):
\begin{equation}\label{eq:shapley_matrix}
SV_{ij}(N,v) = \sum_{i,j \in S \subseteq N} \frac{\Delta_v(S)}{|S|^2}.
\end{equation}
Clearly, \(SV_i(N,v) = \sum_{j \in N} SV_{ij}(N,v)\).
To avoid exponential blow-up, we rely only on the Shapley values and the Shapley value matrix as sources of information about the game and the dependencies between players.

Our procedure operates as follows. 
First, we approximate the Shapley values and the Shapley value matrix by sampling random permutations.
Then, in a loop, we select the player with the highest Shapley value and add 1 to this player's payoff.
To account for this increase, we decrease the player's Shapley value by 1.
However, if a player’s Shapley value is less than 1, we also reduce the Shapley values of other players—primarily those who exhibit synergy with that player. Specifically, if the Shapley value of player $i$ is $\varphi_i < 1$, we decrease the Shapley values of the other players by a total of $(1 - \varphi_i)$. A fraction $\alpha$ of this amount is distributed proportionally to the synergy values $SV_{ij}(N, v)$, and the remaining fraction $(1 - \alpha)$ is distributed equally among them.

\Cref{alg:heuristic} outlines our procedure.
Note that all players will receive their lower quota before we select a player with the Shapley value smaller than 1. 
The subsequent steps mimic the behavior of the indivisible Shapley value: whenever a payoff of 1 is assigned to a player, its contribution to others is removed from the game.
In our algorithm for large games, this contribution is subtracted from the Shapley values of others players. 

To approximate the Shapley value, we use a standard method based on sampling of permutations can be used~\cite{Calvo:etal:1999}.
However, no method exists in the literature to approximate the Shapley value matrix. 
To this end, we propose such a method, depicted in \Cref{alg:matrix_approx}, along with the proof of correctness in \Cref{prop:matrix_approx}.

\begin{algorithm}[t!]
\caption{Approximation of the Shapley value matrix}
\label{alg:matrix_approx}
   
\begin{algorithmic}[1]
\REQUIRE A coalition game $(N,v)$, where $v$ is a black box, number of samples $k$
\ENSURE The Shapley value matrix $SV_{ij}(N,v)$
   \STATE $(\varphi_{ij})_{i,j \in N} \leftarrow 0$ for every $i,j \in N$
   \FOR{$k$ random permutations $\pi$}
      \FOR{$i \in N$}
          \FOR{$j \in S^{\pi}_i$ s.t. $i \le j$}
             \STATE $\varphi_{ij} \leftarrow \varphi_{ij} + \frac{1}{k} (v(S^{\pi}_i \cup \{i\}) - v(S^{\pi}_i) - v(S^{\pi}_i \setminus \{j\} \cup \{i\}) + v(S^{\pi}_i \setminus \{j\})) \cdot \sum_{t=|S^{\pi}_i|+1}^{|N|} \frac{1}{t}$
          \ENDFOR 
      \ENDFOR
   \ENDFOR
   \FOR{$i, j \in N$ s.t. $i < j$}
      \STATE $\varphi_{ji} \leftarrow \varphi_{ij}$
   \ENDFOR
   \STATE {\bfseries return} $(\varphi_{ij})_{i,j \in N}$
\end{algorithmic}
\end{algorithm}

\begin{proposition}\label{prop:matrix_approx}
\Cref{alg:matrix_approx} approximates the Shapley value matrix.
\end{proposition}
\begin{proof}
Fix some $i, j \in N$, where $i \le j$.  
The value $SV_{ij}(N,v)$ based on \cite{Hausken:Mohr:2001} is equivalent to:
\[ SV_{ij}(N,v) \!=\! \sum_{S \subseteq N} \left(\frac{(|S| - 1)!(|N| - |S|)!}{|N|!} \left(v(S) - v(S \setminus \{i\}) - v(S \setminus \{j\}) + v(S \setminus \{i,j\})\right)\sum_{t=|S|}^{|N|} \frac{1}{t}\right). \]

Now, observe that for any $i,j \in S$, the expression $(|S| - 1)!(|N| - |S|)!/|N|!$ represents the probability that, in a random permutation of $|N|$ players, the first $|S|-1$ positions are occupied by players from $S \setminus \{i\}$, and player $i$ appears at position $|S|$.  
Therefore:
\[ SV_{ij}(N,v) = \frac{1}{|N|!} \sum_{\pi \in \Omega(N)} \left((v(S^{\pi}_{i} \cup \{i\}) - v(S^{\pi}_{i}) - v(S^{\pi}_{i} \cup \{i\} \setminus \{j\}) + v(S^{\pi}_{i} \setminus \{j\}))\sum_{t=|S^{\pi}_{i}|+1}^{|N|} \frac{1}{t}\right). \]
It follows that:
\[ SV_{ij}(N,v) = \mathbb{E}_{\pi \sim \Omega(N)} \left((v(S^{\pi}_i \cup \{i\}) - v(S^{\pi}_i) - v(S^{\pi}_i \cup \{i\} \setminus \{j\}) + v(S^{\pi}_i \setminus \{j\}))\sum_{t=|S^{\pi}_i|+1}^{|N|} \frac{1}{t}\right), \]  
where $\mathbb{E}_{\pi \sim \Omega(N)}$ denotes the expected value over $\pi$, which is sampled uniformly from the set of all permutations of $N$ elements.  

It remains to show that, based on the derived formula, a single permutation $\pi$ contributes to and only to those $\varphi_{ij}$ such that $i \le j$, $S^{\pi}_i \cup \{i\}$ is a prefix of $\pi$, and $j \in S^{\pi}_i$, as it is updated in the algorithm. Indeed, this is the case, since when $j \notin S$, we have  
\[ v(S^{\pi}_i \cup \{i\}) - v(S^{\pi}_i) - v(S^{\pi}_i \cup \{i\} \setminus \{j\}) + v(S^{\pi}_i \setminus \{j\}) = 0. \]  
This concludes the proof.
\end{proof}

%%%%%%%%%%%%%%%%%%%%%%%%%%%%%%%%%%%%%%%%%%%%%%%%%%%
\section{Supplementary Materials for \Cref{section:positive}}

%%%%%%%%%%%%%%%%%%%%%%%%%%%%%%%%%%%%%%%%%%%%%%%%%%%

\subsection{Proof of \Cref{theorem:matching}}

%\section{Proofs for \Cref{section:positive}}
%\subsection{Proof of \Cref{theorem:matching}}

\begin{lemma}
\label{lemma:positive_lowerquota}
The allocation $(X_i)_{i \in N}$ returned by \Cref{alg:matching} satisfies Lower Quota and Upper Quota, i.e., it holds $\lfloor SV_i(N,v_{\mathcal{S}}) \rfloor \le |X_i| \le \lceil SV_i(N,v_{\mathcal{S}}) \rceil$ for every player $i \in N$.
\end{lemma}
\begin{proof}
Clearly, the allocation satisfies Upper Quota as it defined by a matching in a graph, in which every player has $\lceil SV_i(N,v) \rceil$ player-nodes.

To prove Lower Quota, it is enough to prove that a perfect matching of player-nodes into objects in line 4 exists.
To this end, we use Hall's marriage theorem. 
For each coalition $S \subseteq N$, the number of objects owned (at least partially) by at least one player in $S$ equals:
\[ |\{o_j \!\in\! O : S \cap S_j \neq \emptyset\}| \geq \sum_{o_j \in O} \frac{|S \!\cap\! S_j|}{|S_j|} = \sum_{i \in S} SV_i(N,v). \]
Thus, for every coalition, the number of object-nodes connected to player-nodes is greater than or equal to the number of these player-nodes.
Hence, the marriage condition is satisfied.
\end{proof}

\begin{lemma}
\label{lemma:positive_efficient}
The allocation $(X_i)_{i \in N}$ returned by \Cref{alg:matching} is efficient, i.e., it is a perfect matching of objects into player-nodes and it holds $\sum_{i \in N} |X_i| = v_\mathcal{S}(N)$.
\end{lemma}
\begin{proof}
We begin by proving that there exists a perfect matching of objects into player-nodes in the final graph, i.e., in a graph in which every player $i$ has $\lceil SV_i(N,v)\rceil$ player-nodes.

Let $M \subseteq \mathcal{S}$ be a set of objects. Let $S$ be the set of players that owns at least partially at least one of the objects from $M$.
The number of player-node of players from $S$ equals:
 \[ \sum_{i \in S} \lceil SV_i(N,v) \rceil \ge \sum_{i \in S} SV_i(N,v) = \sum_{S_j \cap S \neq \emptyset} \frac{|S \cap S_j|}{|S_j|} \ge \sum_{S_j \subseteq S} \frac{|S \cap S_j|}{|S_j|} = |\{S_j : S_j \subseteq S\}| \ge |M|.\]
Thus, for every group of objects, the number of player-nodes connected to the corresponding object-nodes is greater than or equal to the number of these objects.
Hence, Hall's marriage theorem implies the perfect matching exists.

Assume now that \Cref{alg:matching} finds a matching $E'$ that is not perfect.
We know that a perfect matching can be obtained by finding augmenting paths and extending the current matching with them. 
Through such a process all nodes which were matched in the original matching are still matched in the final perfect matching.
Consider a perfect matching $E^*$ obtained through this process from $E'$.
Take a first player $i$ in $\pi$ which has a player-node $v_i$ which is matched in $E^*$, but is not matched in $E'$.
This player-node $v_i$ must have been added in the loop in lines 6--11 of the algorithm, but no augmenting path with this new node was found.

We argue now that it is not possible as evidenced by matching $E^*$.
Consider a graph $G'$ obtained from the final graph by deleting (1) all additional nodes added after player-node $v_i$; (2) all remaining additional nodes unmatched in $E^*$.
Now, $E^*$ restricted to nodes from graph $G'$ is a perfect matching of player-nodes into objects.
In turn, $E'$ restricted to nodes from graph $G'$ is not, as the only unmatched node is $v_i$.
Hence, there exists an augmenting path that contains node $v_i$. 
This path is also an augmenting path in the graph obtained after adding node $v_i$ in the algorithm (this graph differs from $G'$ only by containing additional unmatched player-nodes) which yields a contradiction.
\end{proof}

\begin{algorithm}[t!]
   \caption{Indivisible Shapley Value for Positive Games}
   \label{alg:matching}
   
\begin{algorithmic}[1]
\REQUIRE Objects $O$, players $N$, and a list of object owners $\mathcal{S}$
\ENSURE An allocation of objects $X = (X_1, \dots, X_n)$
   \STATE create a bipartite graph $G_{\mathcal{S}} = (N,O,E)$
   \STATE compute the Shapley value in game $(N, v_\mathcal{S})$
   \STATE duplicate player-node $i$ $\lfloor SV_i(N,v_\mathcal{S}) \rfloor$ times
   \STATE find a perfect matching of player-nodes into objects (using the Hopcroft-Karp algorithm)
   \STATE $\pi \leftarrow$ permutation of players in decreasing order according to $SV_i(N,v_{\mathcal{S}}) - \lfloor SV_i(N,v_{\mathcal{S}}) \rfloor$
   \FOR{$i$ in order $\pi$}
      \IF{$SV_i(N,v_{\mathcal{S}}) > \lfloor SV_i(N,v_{\mathcal{S}}) \rfloor$}
         \STATE add one copy of player-node $i$
         \STATE find an augmenting path with the new node, if exists
	  \ENDIF
   \ENDFOR
   \STATE {\bfseries return} $(\{o_j \in O : \text{ edge } \{i,o_j\} \text{ is selected}\})_{i \in N}$
\end{algorithmic}
\end{algorithm}

\begin{lemma}
\label{lemma:positive_full_core}
For every indivisible payoff vector $x$ that belongs to the core of game $(N,v_{\mathcal{S}})$ there exists an allocation $(X_i)_{i \in N}$ such that $(|X_i|)_{i \in N} = x$.
\end{lemma}
\begin{proof}
Let us create a bipartite graph from $G_{\mathcal{S}}$ by making $x_i$ copies of the player-node for each player $i$.  
Clearly, the number of nodes on both sides is the same and equals $|\mathcal{S}|$.  
It suffices to prove that this graph has a perfect matching.

Let $M \subseteq \mathcal{S}$ be a subset of groups, and let $S$ be the set of players that belong to at least one of these groups.  
The number of player-nodes corresponding to players in $S$ is:
\[ \sum_{i \in S} x_i \geq v(S) = |\{S_j : S_j \subseteq S\}| \geq |M|. \]
Thus, every group of object-nodes is connected to at least the same number of player-nodes.
By Hall's marriage theorem, a perfect matching exists.
\end{proof}

\begin{proof}[Proof of \Cref{theorem:matching}]
Let $(X_i)_{i \in N}$ be an allocation returned by \Cref{alg:matching} and $x = (|X_i|)_{i \in N}$ be the corresponding payoff vector.
We already proved in \Cref{lemma:positive_lowerquota} and \Cref{lemma:positive_efficient} that the payoff vector $x$ satisfies Efficiency, Lower Quota and Upper Quota for game $(N, v_{\mathcal{S}})$.

Also, it is easy to see that $x$ belongs to the core of the game $(N, v_\mathcal{S})$: since every object is assigned to one of its owners, for every coalition $S \subseteq N$ we have: 
\[ \sum_{i \in S} x_i \ge |\{S_j \in \mathcal{S} : S_j \subseteq S\}| = v_{\mathcal{S}}(S). \]

We also proved in \Cref{lemma:lex} that among payoff vectors satisfying the above conditions (satisfying Efficiency, Lower Quota, Upper Quota and belonging to the core) the indivisible Shapley value is lexicographically maximal with respect to permutation $\pi$, where $\pi$ is defined as in \Cref{alg:matching}.
Hence, it remains to argue that \Cref{alg:matching} also finds a payoff vector which is lexicographically maximal.

We will prove this by contradiction: assume it is not true, i.e., there exists a payoff vector $y$ in the core that satisfies Efficiency, Lower Quota and Upper Quota and that is lexicographically maximal.
From \Cref{lemma:positive_full_core} we know that there exists an allocation that corresponds to this payoff vector.
Now, consider first player $i$ (with respect to $\pi$) on which both payoff vector differ. 
Since both vectors satisfy Lower Quota and Upper Quota we get that $x_i = \lfloor SV_i(N,v) \rfloor$ and $y_i = \lceil SV_i(N,v) \rceil$ and $x_i < y_i$.
Now, consider the graph $G_i$ obtained by adding player-node for player $i$ in the loop in lines 6--10 of \Cref{alg:matching}.
Since there exists an allocation that corresponds to $y$ we get that there exists a larger matching in graph $G_i$.
This implies that there exists an augmenting path.
This augmenting path must contain the new copy of the player-node $i$ as there was no augmenting path before its addition---if it was, then $y$ is not lexicographically maximal.
Hence, \Cref{alg:matching} would find this augmenting path and assign $\lceil SV_i(N,v) \rceil$ to player $i$ which is a contradiction.

It remains to analyze the complexity of \Cref{alg:matching}.
Creating a bipartite graph $G_{\mathcal{S}}$ and computing the Shapley value (lines 1--2) takes time $O(|N||\mathcal{S}|)$.
The graph obtained from the duplication of the nodes has $O(|\mathcal{S}|+|N|)$ nodes (the sum of the Shapley values equal $v(N) = |\mathcal{S}|$) and $O(|\mathcal{S}|(|\mathcal{S}|+|N|))$ edges.
Hence, line 3 takes time $O(|\mathcal{S}|(|\mathcal{S}|+|N|))$.
The Hopcroft-Karp algorithm takes time $O(|\mathcal{S}|(|\mathcal{S}|+|N|)^{3/2})$ (line 4).
Then, the main loop (lines 6--10) has at most $|N|$ iterations and finding the augmenting path takes time $O(|\mathcal{S}|(|\mathcal{S}|+|N|))$.
Hence, the full cost of the algorithm is $O(|\mathcal{S}|(|\mathcal{S}|+|N|)^{3/2} + |\mathcal{S}| |N| (|\mathcal{S}|+|N|))$.
If $|N| \le |\mathcal{S}|$, this simplifies to $O(|\mathcal{S}|^{5/2})$.
\end{proof}

\begin{proposition}
\label{proposition:positive_polynomial}
The indivisible Shapley value of a positive game can be calculated in polynomial time with respect to the number of non-zero Harsanyi dividends and players.
\end{proposition}
\begin{proof}[Proof of \Cref{proposition:positive_polynomial}]
Consider an arbitrary positive game with $n$ players and $d$ non-zero dividends.
If $\Delta_v(S) \ge |S|$ for some coalition $S$, then we assign $\lfloor \Delta_v(S)/|S| \rfloor$ to each player $i \in S$. 
Then, we create a list $\mathcal{S}$ in which every coalition $S$ appears $\Delta_v(S) - |S| \cdot \lfloor \Delta_v(S)/|S| \rfloor$ times.
The size of the list does not exceed $d \cdot n$.
Now, by running \Cref{alg:matching} we get an allocation $(X_i)_{i \in N}$ that corresponds to the indivisible Shapley value for game $v_\mathcal{S}$. 
By adding to the payoff of each player $i$ value $|X_i|$ we get the indivisible Shapley value for the original game.
\end{proof}

\section{Supplementary Materials for \Cref{section:case_studies}}

\subsection{Additional Experiments for \Cref{section:image_classification}}
\label{experiment:statistics}

% Average ratio of shapley: 0.673 ± 0.147
% Average ratio of synergy: 0.107  ± 0.188
% Average number of clusters for alpha=0.0: 2.703  ± 0.960
% Average number of clusters for alpha=0.5: 3.228 ± 1.061
% Average silhouette score for alpha=0.0: 0.583  ± 0.105
% Average silhouette score for alpha=0.5: 0.505  ± 0.077

% Statistical Tests (p-values):

% One-Sided Tests (vs constant):
% Shapley Ratio > 0.68: p = 0.6778
% Synergy Ratio < 0.16: p = 0.0029

% Two-Sided Tests:
% Comparing number of clusters (alpha=0.0 vs alpha=0.5): p = 0.0003
% Comparing silhouette scores (alpha=0.0 vs alpha=0.5): p = 0.0000
We conducted the same experiments for random 100 images from the dataset with their ground truth labels. Each image was divided using skimage.segmentation.slic with the parameter n\_segments equal to 250. We approximated each image's Shapley value matrix with 5,000 samples. Similarly, we selected $v(N) = 10$ regions. We compared our method with $\alpha=0.5$ and baseline approach that selects regions with the highest Shapley value.

On average, the sum of Shapley values of features selected by our method was lower only by $26.7 \pm 14.7 \%$ compared to the the baseline method, while their sum of synergies was lower by $89.3 \pm 18.8\%$. Thus, our method sacrifices only a small amount of Shapley value while achieving a substantial reduction in synergy among the selected features.

Next, we applied Euclidean k-means clustering (100 runs for each $k \in [2, 9]$) and selected the solution with the highest Silhouette Score. The average number of clusters for our method was $3.228 \pm 1.061$ compared to $2.703 \pm 0.960$ for the baseline. The average highest Silhouette Score for our method was $0.505  \pm 0.077$ compared to $0.583  \pm 0.105$ for the baseline. 

We assume that the Silhouette Scores for both methods are normally distributed, as the p-values of the Shapiro--Wilk tests are greater than 0.25 for both. Therefore, we cannot reject the null hypothesis that the data come from a normal distribution. We then performed a one-sided Welch's $t$-test to compare the mean Silhouette Scores, with the alternative hypothesis that our method yields lower scores than the baseline. The resulting p-value of $p < 10^{-8}$ indicates a statistically significant difference showing that our method has a lower average Silhouette Score than the baseline method.

\subsection{Auxiliary Visualization for \Cref{section:image_classification}}

\label{visualization:zoo}
In \Cref{figure:zoo} we showcase selection of key regions in an image classification for different $\alpha$ parameter (Used in \cref{alg:matrix_approx}). We picked four extra images with labels \textit{ice\_bear}, \textit{elephant}, \textit{tiger}, \textit{tiger\_cat}. 

\begin{figure}[H]

\begin{tikzpicture}[x=1cm, y=1cm]
\node at (-13.0, 3.0) {\small Shapley Value};
\node at (-13.0, 0.0) {\small $\alpha = 0.0$};
\node at (-13.0, -3.0) {\small $\alpha = 0.2$};
\node at (-13.0, -6.0) {\small $\alpha = 0.4$};
\node at (-13.0, -9.0) {\small $\alpha = 0.6$};
\node at (-13.0, -12.0) {\small $\alpha = 0.8$};
\node at (-13.0, -15.0) {\small $\alpha = 1.0$};

\begin{scope}[shift={(-4,0)}]
\node at (-6.0, 3.0) {{\includegraphics[width=0.2\textwidth]{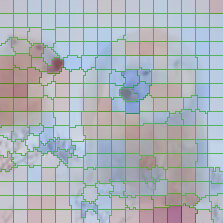}}};
\node at (-6.0, 0.0) {{\includegraphics[width=0.2\textwidth]{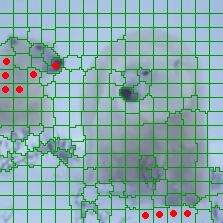}}};
\node at (-6.0, -3.0) {{\includegraphics[width=0.2\textwidth]{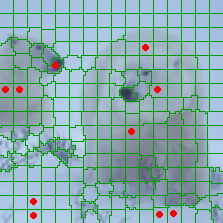}}};
\node at (-6.0, -6.0) {{\includegraphics[width=0.2\textwidth]{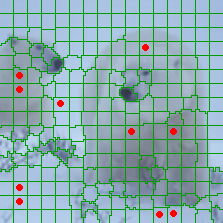}}};
\node at (-6.0, -9.0) {{\includegraphics[width=0.2\textwidth]{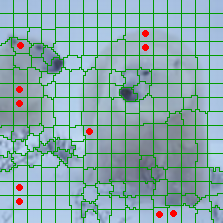}}};
\node at (-6.0, -12.0) {{\includegraphics[width=0.2\textwidth]{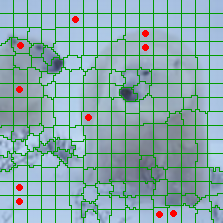}}};
\node at (-6.0, -15.0) {{\includegraphics[width=0.2\textwidth]{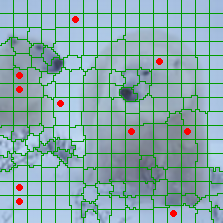}}};

\node at (-3.0, 3.0) {{\includegraphics[width=0.2\textwidth]{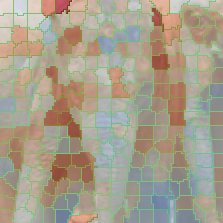}}};
\node at (-3.0, 0.0) {{\includegraphics[width=0.2\textwidth]{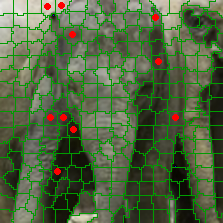}}};
\node at (-3.0, -3.0) {{\includegraphics[width=0.2\textwidth]{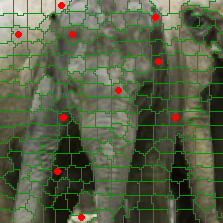}}};
\node at (-3.0, -6.0) {{\includegraphics[width=0.2\textwidth]{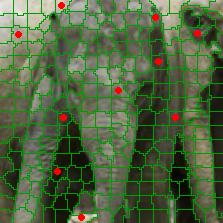}}};
\node at (-3.0, -9.0) {{\includegraphics[width=0.2\textwidth]{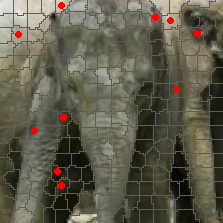}}};
\node at (-3.0, -12.0) {{\includegraphics[width=0.2\textwidth]{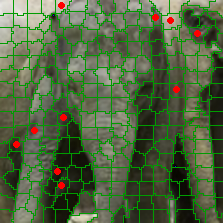}}};
\node at (-3.0, -15.0) {{\includegraphics[width=0.2\textwidth]{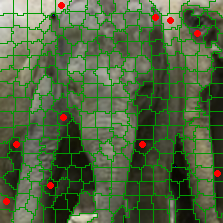}}};

\node at (-0.0, 3.0) {{\includegraphics[width=0.2\textwidth]{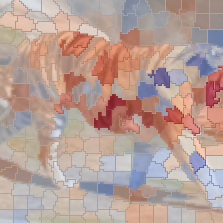}}};
\node at (-0.0, 0.0) {{\includegraphics[width=0.2\textwidth]{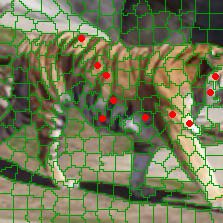}}};
\node at (-0.0, -3.0) {{\includegraphics[width=0.2\textwidth]{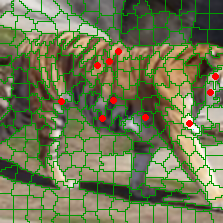}}};
\node at (-0.0, -6.0) {{\includegraphics[width=0.2\textwidth]{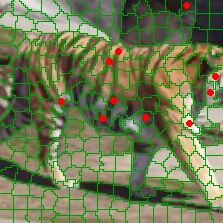}}};
\node at (-0.0, -9.0) {{\includegraphics[width=0.2\textwidth]{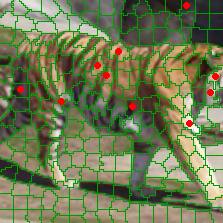}}};
\node at (-0.0, -12.0) {{\includegraphics[width=0.2\textwidth]{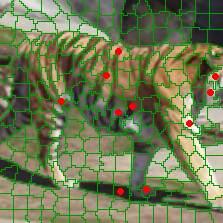}}};
\node at (-0.0, -15.0) {{\includegraphics[width=0.2\textwidth]{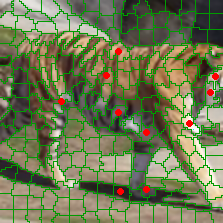}}};

\node at (3.0, 3.0) {{\includegraphics[width=0.2\textwidth]{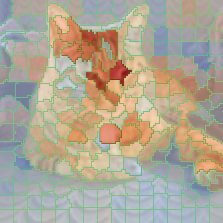}}};
\node at (3.0, 0.0) {{\includegraphics[width=0.2\textwidth]{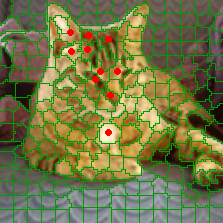}}};
\node at (3.0, -3.0) {{\includegraphics[width=0.2\textwidth]{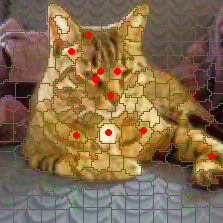}}};
\node at (3.0, -6.0) {{\includegraphics[width=0.2\textwidth]{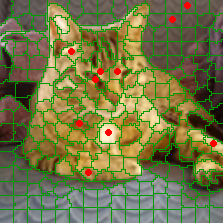}}};
\node at (3.0, -9.0) {{\includegraphics[width=0.2\textwidth]{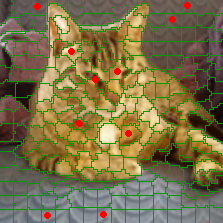}}};
\node at (3.0, -12.0) {{\includegraphics[width=0.2\textwidth]{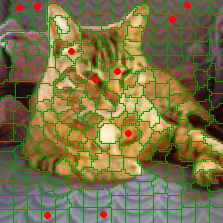}}};
\node at (3.0, -15.0) {{\includegraphics[width=0.2\textwidth]{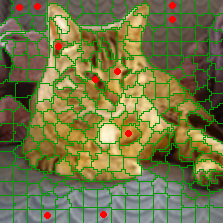}}};

\end{scope}

\end{tikzpicture}
\caption{Comparison of our method for different values of parameter $\alpha$  on four pictures from MiniImageNet. Red dots indicate key regions chosen by our method. Color scale on figures with Shapley Value goes from red for regions with highest value to blue for regions with lowest value. Scale is centered at 0.}
\label{figure:zoo}
\end{figure}
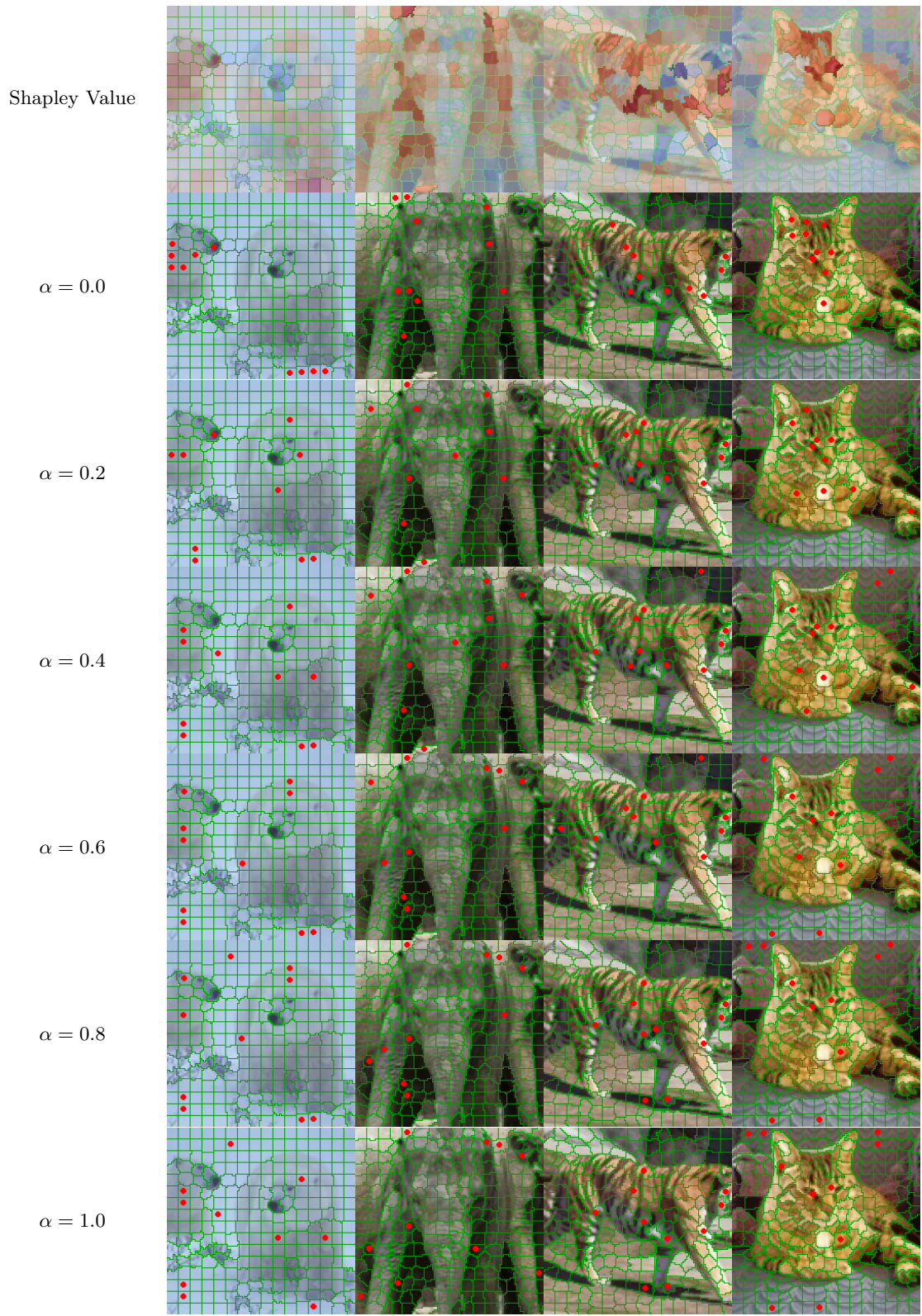

\end{document}